\documentclass[a4paper,11pt,reqno]{amsart}
\usepackage[body={15.5cm,20cm}]{geometry}
\usepackage{amssymb}
\usepackage{mathtools}
\usepackage{microtype}
\usepackage{enumerate}
\usepackage{graphicx}
\numberwithin{equation}{section}
\theoremstyle{plain}
\newtheorem{thm}{Theorem}[section]
\newtheorem{prop}[thm]{Proposition}
\newtheorem*{thm*}{Theorem}
\theoremstyle{remark}
\newtheorem{rem}[thm]{Remark}
\providecommand{\D}{\mathbb}
\newcommand{\dd}{\mathrm{d}}
\newcommand{\ee}{\mathrm{e}}
\newcommand{\ii}{\mathrm{i}}
\providecommand{\abs}[1]{\lvert#1\rvert}
\providecommand{\accol}[1]{\lbrace#1\rbrace}
\providecommand{\croch}[1]{\lbrack#1\rbrack}
\renewcommand{\Im}{\operatorname{Im}}
\newcommand{\ord}{\mathrm{O}}
\newcommand{\osmall}{\mathrm{o}}
\renewcommand{\Re}{\operatorname{Re}}
\DeclareMathOperator{\diag}{diag}
\DeclareMathOperator{\one}{I}
\DeclareMathOperator{\Res}{Res}
\DeclareMathOperator{\two}{II}
\DeclareMathOperator{\three}{III}
\begin{document}
\title[RHP for Degasperis--Procesi]{A Riemann--Hilbert approach\\
for the Degasperis--Procesi equation}
\author[A.~Boutet de Monvel]{Anne Boutet de Monvel$^{\ast}$}
\author[D.~Shepelsky]{Dmitry Shepelsky$^{\dagger}$}
\address{$^{\ast}$%
Institut de Math\'ematiques de Jussieu-PRG,
Universit\'e Paris Diderot,
Avenue de France,
B\^at.~Sophie Germain,
case 7012,
75205 Paris Cedex 13,
France}
\email{aboutet@math.jussieu.fr}
\address{$^{\dagger}$%
Mathematical Division,
Institute for Low Temperature Physics,
47 Lenin Avenue,
61103 Kharkiv,
Ukraine}
\email{shepelsky@yahoo.com}
\date{\today}
\begin{abstract}
We present an inverse scattering transform approach to the Cauchy problem on the line for the Degasperis--Procesi equation 
\[
u_t-u_{txx}+3\omega u_x+4uu_x=3u_xu_{xx}+uu_{xxx}
\]
in the form of an associated Riemann--Hilbert problem. This approach allows us to give a representation of the  solution to the Cauchy problem, which can be efficiently used in studying its long-time behavior.
\end{abstract}
\maketitle
\section{Introduction}                   \label{sec:intro}

In this paper we present the inverse scattering approach, based on an appropriate Riemann--Hilbert problem formulation, for the initial value problem for the Degasperis-Procesi (DP) equation \cite{DP99,DHH02}
\begin{align}                                               \label{DP-om}
&u_t-u_{txx}+3\omega u_x+4uu_x=3u_xu_{xx}+uu_{xxx},\quad -\infty<x<\infty,\ t>0,\\
&u(x,0)=u_0(x), \label{IC}
\end{align}
where $\omega$ is a positive parameter. The DP equation arises as a model equation describing the shallow-water approximation in inviscid hydrodynamics in the so-called ``moderate amplitude regime'': introducing two small parameters, the wave-amplitude parameter $\varepsilon$ (characterizing the smallness of the wave amplitude) and the long-wave parameter $\delta$ (characterizing the smallness of the typical wavelength with respect to the water depth), in this regime we assume that  $\delta\ll 1$ and $\varepsilon\sim\delta$. This regime can be characterized as to be more nonlinear than dispersive, which, in particular, allows ``wave breaking''. This  is in contrast with the so-called ``shallow water regime'' ($\delta\ll 1$ and $\varepsilon\sim\delta^2$), where nonlinearity and dispersion are so balanced that the solution of the initial value problem for the associated nonlinear equation (the Korteweg--de Vries equation) exists globally for all times, for all ``nice'' (sufficiently decaying and smooth) initial data.

Among the models of moderate amplitude regime, only two are integrable (admitting a bi-Hamiltonian structure and a Lax pair representation): they are the Camassa--Holm (CH) equation and the DP equation. Also, they are the only two integrable equations from the ``$b$-family'' of equations
\[
u_t-u_{txx}+ b \omega u_x+(b+1)uu_x=b u_xu_{xx}+uu_{xxx}.
\]
The CH and DP equations correspond to $b=2$ and $b=3$, respectively.

The analysis of the CH equation by using the inverse scattering approach was started in \cite{C01,CL03,L02}. A version of the inverse scattering method based on the Riemann--Hilbert (RH) factorization problem was proposed in \cite{BS06,BS08} (another Riemann--Hilbert formulation of the inverse scattering transform is presented in \cite{CGI06}). The RH approach has proved its efficiency in the study of the long-time behavior of solutions of both initial value problems \cite{BS-D,BKST,BIS10} and initial boundary value problems \cite{BS-F}. 
 
In the present paper we develop the Riemann--Hilbert approach to the DP equation, following the main ideas developed in \cite{BS08}. 
 
A major difference between the implementations of the Riemann--Hilbert method to the CH equation and the DP equation is that in the latter case, the spatial equation of the Lax pair is of the third order, which implies that in the matrix form one has to deal with $3\times 3$ matrix-valued equations, while in the case of the CH equation, they have a  $2\times 2$ matrix structure, as in the cases of the most known integrable equations (KdV, modified KdV, nonlinear Schr\"odinger, sine--Gordon, etc.) Hence, the construction and analysis of the associated RH problem become considerably more complicated.

In a recent paper \cite{CIL10}, the inverse scattering method for the DP equation based on a $3\times 3$ matrix RH problem in the spectral $k$-plane is proposed, where the solution of the DP equation is extracted from the large-$k$ behavior of the solution of the RH problem. Also, the dressing procedure is given for constructing $N$-soliton solutions, which is illustrated, particularly, by the explicit construction (in parametric form) of the $1$-soliton solutions. 
 
In our approach, we propose a different Riemann--Hilbert problem and give a different representation of the solution $u(x,t)$ of the initial value problem \eqref{DP-om} in terms of the solution of this RH problem evaluated at a distinguished point of the plane of the spectral parameter. Remarkably, the  formulae for $u(x,t)$ have the same structure as the parametric formulae for pure multisoliton solutions obtained in \cite{M05}.
 
\section{Lax pair and eigenfunctions}

We assume that the initial function $u_0(x)$ in \eqref{IC} is sufficiently smooth and  decay fast as $|x|\to\infty$. Letting $u(x,t)$ be the solution of the DP equation, we introduce the ``momentum'' variable 
\[
m=m(x,t)\equiv u-u_{xx}.
\]
It is known that, similarly to the case of the CH equation (see, e.g., \cite{C01}), the condition $m(x,0)+\omega>0$ for all $x$ provides the existence of $m(x,t)$  for all $t$; moreover, $m(x,t)+\omega>0$ for all $x\in\D{R}$ and all $t>0$. This justifies the form 
\begin{equation}\label{DP-b}
\left(
\sqrt[3]{m+\omega}\right)_t=- \left(u \sqrt[3]{m+\omega}\right)_x
\end{equation}
 of the DP equation, which will be used in our constructions below.
 
The  linear dispersion parameter $\omega$ can be scaled out to $\omega=1$. Hence, in what follows, for simplicity, we assume $\omega=1$.
 
\subsection{Lax pairs}   \label{lax.pairs}

\subsubsection*{Scalar Lax pair}

The DP equation admits a Lax representation: this equation is actually the compatibility condition of two linear equations \cite{DHH02}
\begin{subequations}\label{Lax-ini}
\begin{align} \label{Lax-ini-x}
\psi_{x}-\psi_{xxx}&=-z^3(m+1)\psi,\\  \label{Lax-ini-t}
\psi_t &= \Bigl(u_x-\frac{c}{z^3}\Bigr)\psi-u\psi_x +\frac{1}{z^3}\psi_{xx},
\end{align}
\end{subequations}
where $z$ is a spectral parameter, $\psi=\psi(x,t,z)$, and $c$ is an arbitrary constant. In what follows we will see that it is convenient to choose $c=\frac{2}{3}$.

\subsubsection*{1st matrix form}
In order to control the behavior of solutions to \eqref{Lax-ini} as functions of the spectral parameter $z$ (which is crucial for the Riemann--Hilbert method), it is convenient to rewrite the Lax pair in matrix form. Introducing $\Phi=\Phi(x,t,z)$ by
\[
\Phi=\begin{pmatrix}\psi\\ \psi_x \\ \psi_{xx}\end{pmatrix}
\]
transforms the scalar Lax pair \eqref{Lax-ini} into a Lax pair
\begin{subequations} \label{Lax-vec}
\begin{align} \label{Lax-vec-x}
\Phi_x &= U \Phi,\\ 
\label{Lax-vec-t}
\Phi_t &= V \Phi,
\end{align} 
where 
\begin{align}\label{U}
&U(x,t,z)=\begin{pmatrix}
0 & 1 & 0 \\
0 & 0 & 1 \\
z^3 q^3 & 1 & 0
\end{pmatrix},\\
\label{V}
&V(x,t,z)=\begin{pmatrix}
u_x-\frac{c}{z^3} & -u & \frac{1}{z^3} \\
u+1 & -\frac{c-1}{z^3} & -u \\
u_x-z^3 u q^3 & 1 & -u_x -\frac{c-1}{z^3}
\end{pmatrix},
\end{align}
with
\begin{equation}  \label{q}
q=q(x,t)=(m+1)^{1/3}. 
\end{equation}
\end{subequations}
From now on, \eqref{Lax-vec} will be seen as a $3\times 3$ matrix-valued linear system: a ``matrix'' solution of \eqref{Lax-vec} is a collection $\Phi=(\Phi^{(1)}\ \Phi^{(2)}\ \Phi^{(3)})$ of three ``vector'' solutions $\Phi^{(j)}$. Now we notice that if $c=\frac{2}{3}$, then $V$ is traceless. Thus in this case the determinant of a matrix solution to the equation $\Phi_t = V \Phi$ is independent of $t$. The analogous property of $\Phi_x = U \Phi$ is obvious.

The coefficient matrices $U$ and $V$ in \eqref{Lax-vec} have singularities (in the extended complex $z$-plane) at $z=0$ and at $z=\infty$. In order to control the large-$z$ behavior of solutions of \eqref{Lax-vec}, we follow a strategy similar to that adopted for the CH equation \cite{BS06, BS08}. We transform \eqref{Lax-vec} in such a way that: 
\begin{enumerate}[i)]
\item
the leading terms for $z\to\infty$ in the Lax equations be diagonal, whereas the terms of order $\ord(1)$ be off-diagonal; 
\item
all lower order terms (including those of order $\ord(1)$) vanish as $|x|\to \infty$. 
\end{enumerate}
It is instructive to perform this transformation in two steps:
\begin{enumerate}[(i)]
\item
First, transform \eqref{Lax-vec} into a system where the leading terms are represented as 
products of $(x,t)$-independent and $(x,t)$-dependent factors.
\item
Second, diagonalize the $(x,t)$-independent factors.
\end{enumerate}

\subsubsection*{2nd matrix form}

For the first step, introducing $\tilde\Phi=\tilde\Phi(x,t,z)$ by
\[
\tilde\Phi=D^{-1}\Phi,
\]
where
\[
D(x,t)=\diag\{q^{-1}(x,t),1,q(x,t)\}
\]
transforms \eqref{Lax-vec} into a new Lax pair
\begin{subequations} \label{Lax-vec-1}
\begin{align} \label{Lax-vec-1-x}
\tilde\Phi_x &= \tilde U \tilde\Phi,\\
\label{Lax-vec-1-t} 
\tilde\Phi_t &= \tilde V \tilde\Phi,
\end{align} 
where 
\begin{align}  \label{U-tilde}
\tilde U(x,t,z)
&=q(x,t)\begin{pmatrix}
0 & 1 & 0 \\
0 & 0 & 1 \\
z^3  & 1 & 0
\end{pmatrix} + \begin{pmatrix}
\frac{q_x}{q} & 0 & 0 \\
0 & 0 & 0 \\
0  & \frac{1}{q}-q & -\frac{q_x}{q}
\end{pmatrix}\notag\\
&\equiv q(x,t) U_\infty (z) + \tilde U^{(1)}(x,t)
\end{align}
and
\begin{align}
\tilde V(x,t,z)
&= - u q \begin{pmatrix}
0 & 1 & 0 \\
0 & 0 & 1 \\
z^3  & 1 & 0
\end{pmatrix} + \begin{pmatrix}
-\frac{2}{3z^3} & 0 & \frac{1}{z^3} \\
1 & \frac{1}{3 z^3} & 0 \\
0  & 1 & \frac{1}{3 z^3}
\end{pmatrix}\notag\\
&\quad + \begin{pmatrix}
-u\frac{q_x}{q} & 0 & 0 \\
\frac{u+1}{q}-1 & 0 & 0 \\
\frac{u_x}{q^2} & \frac{1}{q}-1+uq & u\frac{q_x}{q}
\end{pmatrix} + \frac{q^2-1}{z^3}\begin{pmatrix}
0 & 0 & 1 \\
0 & 0 & 0 \\
0  & 0 & 0
\end{pmatrix}\notag\\ 
\label{V-tilde}
&\equiv -u(x,t)q(x,t) U_\infty (z) + V_\infty (z) + \tilde V^{(1)}(x,t) + \frac{1}{z^3}\tilde V^{(2)}(x,t).
\end{align}
\end{subequations}

\subsubsection*{Main matrix form}

For the second step, it is important that the commutator of $U_\infty$ and $V_\infty$ vanishes identically, i.e., $[U_\infty, V_\infty]\equiv 0$, which allows simultaneous diagonalization of $U_\infty$ and $V_\infty$. Indeed, we have
\begin{subequations}\label{la-a}
\begin{align}\label{lambda}
P^{-1}(z)U_\infty (z)P(z)&=\Lambda(z),\\ 
\label{a}
P^{-1}(z)V_\infty (z)P(z)&=A(z),
\end{align}
\end{subequations}
where $I$ is the identity $3\times 3$ matrix,
\begin{subequations} \label{La-P}
\begin{align}\label{La}
\Lambda(z) &= \begin{pmatrix}
\lambda_1(z) & 0 & 0 \\
0 & \lambda_2(z) & 0 \\
0 & 0 & \lambda_3(z)
\end{pmatrix},\\
\label{Aa}
A(z)&=\frac{1}{3z^3}I+\Lambda^{-1}(z) 
\equiv
\begin{pmatrix}
A_1(z) & 0 & 0 \\
0 &A_2(z) & 0 \\
0 & 0 &A_3(z)
\end{pmatrix},\\
\label{P}
P(z) &= \begin{pmatrix}
 1 & 1 & 1 \\
\lambda_1(z) & \lambda_2(z) & \lambda_3(z) \\
\lambda_1^2(z) & \lambda_2^2(z) & \lambda_3^2(z)
\end{pmatrix},
\end{align}
and
\begin{equation} \label{La-P-1}
P^{-1}(z) = \left(\begin{smallmatrix}
 (3\lambda_1^2(z)-1)^{-1} & 0 & 0 \\
0 & (3\lambda_2^2(z)-1)^{-1} & 0 \\
0  & 0 & (3\lambda_3^2(z)-1)^{-1}
\end{smallmatrix}\right)
\left(\begin{smallmatrix}
 \lambda_1^2(z)-1 & \lambda_1(z) & 1 \\
 \lambda_2^2(z)-1 & \lambda_2(z) & 1\\
 \lambda_3^2(z)-1 & \lambda_3(z) & 1
\end{smallmatrix}\right).
\end{equation}
\end{subequations}
Here $\lambda_j(z)$, $j=1,2,3$ are the solutions of the algebraic equation 
\begin{equation}\label{la-eqn}
\lambda^3 - \lambda=z^3,
\end{equation}
so that $\lambda_j(z)\sim\omega^jz$ as $z\to\infty$, where $\omega=\ee^{\frac{2\ii\pi}{3}}$.

Thus, introducing $\hat\Phi=\hat\Phi(x,t,z)$ by
\[
\hat\Phi=P^{-1}\tilde\Phi
\]
leads to another Lax pair
\begin{subequations} \label{Lax-vec-2}
\begin{align} 
&\hat\Phi_x - q\Lambda(z)\hat\Phi = \hat U \hat\Phi,\\ 
&\hat\Phi_t + (uq\Lambda(z) - A(z))\hat\Phi = \hat V \hat\Phi,
\end{align} 
where 
\begin{align}
\hat U(x,t,z)&= P^{-1}(z) \tilde U^{(1)}(x,t) P(z),\\
\hat V(x,t,z)&= P^{-1}(z)\left(\tilde V^{(1)}(x,t)+\frac{1}{z^3}\tilde V^{(2)}(x,t)\right)P(z).
\end{align}
\end{subequations}

\subsubsection*{Commutator form}

Notice that $\hat U(x,t,z)=\ord(1)$ and $\hat V(x,t,z)=\ord(1)$ as $z\to\infty$ since $\tilde U^{(1)}$ and $\tilde V^{(1)}$ are lower triangular matrices. Moreover, it can be  checked directly that the diagonal entries of $\hat U(x,t,z)$ and $\hat V(x,t,z)$ are of order $\ord(1/z)$. This latter fact is important to establish the large-$z$ behavior of $\hat\Phi$ \cite{BC}.

On the other hand, $\hat U(x,t,z)=\osmall(1)$ and $\hat V(x,t,z)=\osmall(1)$ as $|x|\to\infty$, which suggests introducing a $3\times 3$ diagonal function $Q(x,t,z)$ that solves the system 
\begin{subequations} \label{Q-dif}
\begin{align} \label{Q-dif-x}
Q_x &=q\Lambda(z),\\ 
\label{Q-dif-t}
Q_t &=-uq\Lambda(z) + A(z),
\end{align}
\end{subequations}
as follows:
\begin{equation}\label{Q}
Q(x,t,z)=y(x,t)\Lambda(z)+tA(z)
\end{equation}
with
\begin{equation}\label{y}
y(x,t)=x-\int_x^\infty (q(\xi,t)-1)\dd\xi.
\end{equation}
The fact that the equations in \eqref{Q-dif} are consistent follows directly from the DP equation in the form \eqref{DP-b}: $q_t=-(uq)_x$. The normalization of $Q(x,t,z)$ is chosen in such a way that 
\begin{equation}\label{Q-inf}
Q(x,t,z)\sim x\Lambda(z)+tA(z)\quad\text{as }x\to+\infty.
\end{equation}
The role of $Q(x,t,z)$ in the construction of the Riemann--Hilbert problem is to 
catch the large-$z$ behavior of dedicated solutions of the Lax pair equations \eqref{Lax-vec-2}. Indeed, introducing the $3\times 3$ matrix-valued function $M=M(x,t,z)$ by
\[
M=\hat\Phi \ee^{-Q}
\]
reduces \eqref{Lax-vec-2}
to the system 
\begin{subequations} \label{M}
\begin{align} \label{M-x}
M_x - [Q_x, M] &= \hat U M,\\ 
\label{M-t}
M_t - [Q_t, M] &= \hat V M.
\end{align} 
\end{subequations}

\subsection{Eigenfunctions}   \label{eigenfunctions}

\subsubsection*{Fredholm integral equations}

Particular solutions of \eqref{M} having well-controlled properties as functions of the spectral parameter $z$ can be constructed as solutions of the Fredholm integral equation (cf.~\cite{BC})
\begin{align} \label{M-int}
&M(x,t,z)=\notag\\
&I + \int_{(x^*,t^*)}^{(x,t)} \ee^{Q(x,t,z)-Q(\xi,\tau,z)}
	\left( \hat U M(\xi,\tau,z)\dd\xi + \hat V M(\xi,\tau,z)\dd\tau\right)\ee^{-Q(x,t,z)+Q(\xi,\tau,z)},
\end{align}
where the initial points of integration $(x^*,t^*)$ can be chosen differently for different matrix entries of the equation: $Q$ being diagonal, \eqref{M-int} must be seen as the collection of scalar integral equations ($1\leq j,l\leq 3$)
\begin{align*}
&M_{jl}(x,t,z)=\\
&I_{jl}+\int_{(x_{jl}^*,t_{jl}^*)}^{(x,t)}\ee^{Q_{jj}(x,t,z)-Q_{jj}(\xi,\tau,z)}\bigl((\hat UM)_{jl}(\xi,\tau,z)\dd\xi+(\hat VM)_{jl}(\xi,\tau,z)\dd\tau\bigr)\ee^{-Q_{ll}(x,t,z)+Q_{ll}(\xi,\tau,z)}.
\end{align*}
Choosing the $(x_{jl}^*,t_{jl}^*)$ appropriately allows to obtain eigenfunctions that can be used in the construction of the Riemann--Hilbert problems associated with the initial value problems \cite{BS08} as well as the initial boundary value problems \cite{BS-F}. 
\begin{figure}[ht]
\centering\includegraphics[scale=1]{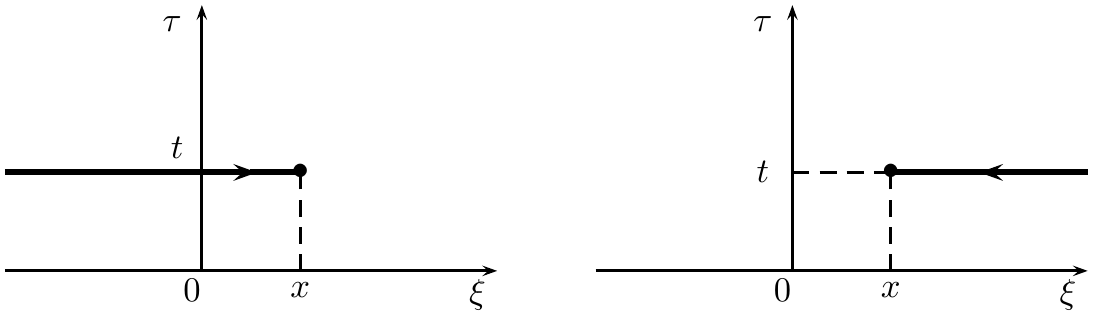}
\caption{Paths of integration. Left: $(x^*,t^*)=(-\infty,t)$. Right: $(x^*,t^*)=(+\infty,t)$.} 
\label{fig:paths}
\end{figure}

In particular, for the Cauchy problem considered in the present paper, it is reasonable to choose these points to be $(-\infty,0)$ or $(+\infty,0)$ thus reducing the integration in \eqref{M-int} to paths parallel to the $x$-axis (see Figure~\ref{fig:paths}) provided the integrals $\int_{(-\infty,0)}^{(-\infty,t)}$ and $\int_{(\infty,0)}^{(\infty,t)}$ vanish:
\begin{equation} \label{M-int1}
M(x,t,z)=I+\int_{(\pm)\infty}^x\ee^{Q(x,t,z)-Q(\xi,t,z)}\croch{\hat UM(\xi,t,z)}\ee^{-Q(x,t,z)+Q(\xi,t,z)}\dd\xi,
\end{equation}
or, in view of \eqref{Q},
\begin{equation} \label{M-int2}
M(x,t,z)=I+\int_{(\pm)\infty}^x\ee^{-(\int_x^{\xi}q(\zeta,t)\dd\zeta)\Lambda(z)}\croch{\hat UM(\xi,t,z)}\ee^{(\int_x^{\xi}q(\zeta,t)\dd\zeta)\Lambda(z)}\dd\xi.
\end{equation}
Since $q>0$, the domains (in the complex $z$-plane) where the exponential factors in \eqref{M-int1} are bounded are determined by the signs of the differences $\Re\lambda_j(z)-\Re\lambda_l(z)$, $1\leq j\neq l\leq 3$.

\subsubsection*{A new spectral parameter}
It is convenient to introduce a new spectral parameter $k$ (see \cite{CIL10} and also \cite{Ca82}) such that
\begin{equation} \label{z-k}
z(k)=\frac{1}{\sqrt{3}}\,k \left(1+\frac{1}{k^6}\right)^{1/3},
\end{equation}
and fixed by the condition $z(k)\sim\frac{1}{\sqrt{3}}\,k$ as $k\to\infty$. We fix $\lambda_j=\lambda_j(k)$ as follows
\begin{equation} \label{la-k}
\lambda_j(k)=\frac{1}{\sqrt{3}}  \left(\omega^j k +\dfrac{1}{\omega^j k}\right)\ \text{ where }\ \omega=\ee^{\frac{2\ii\pi}{3}}.
\end{equation}
\begin{figure}[ht]
\centering\includegraphics[scale=.8]{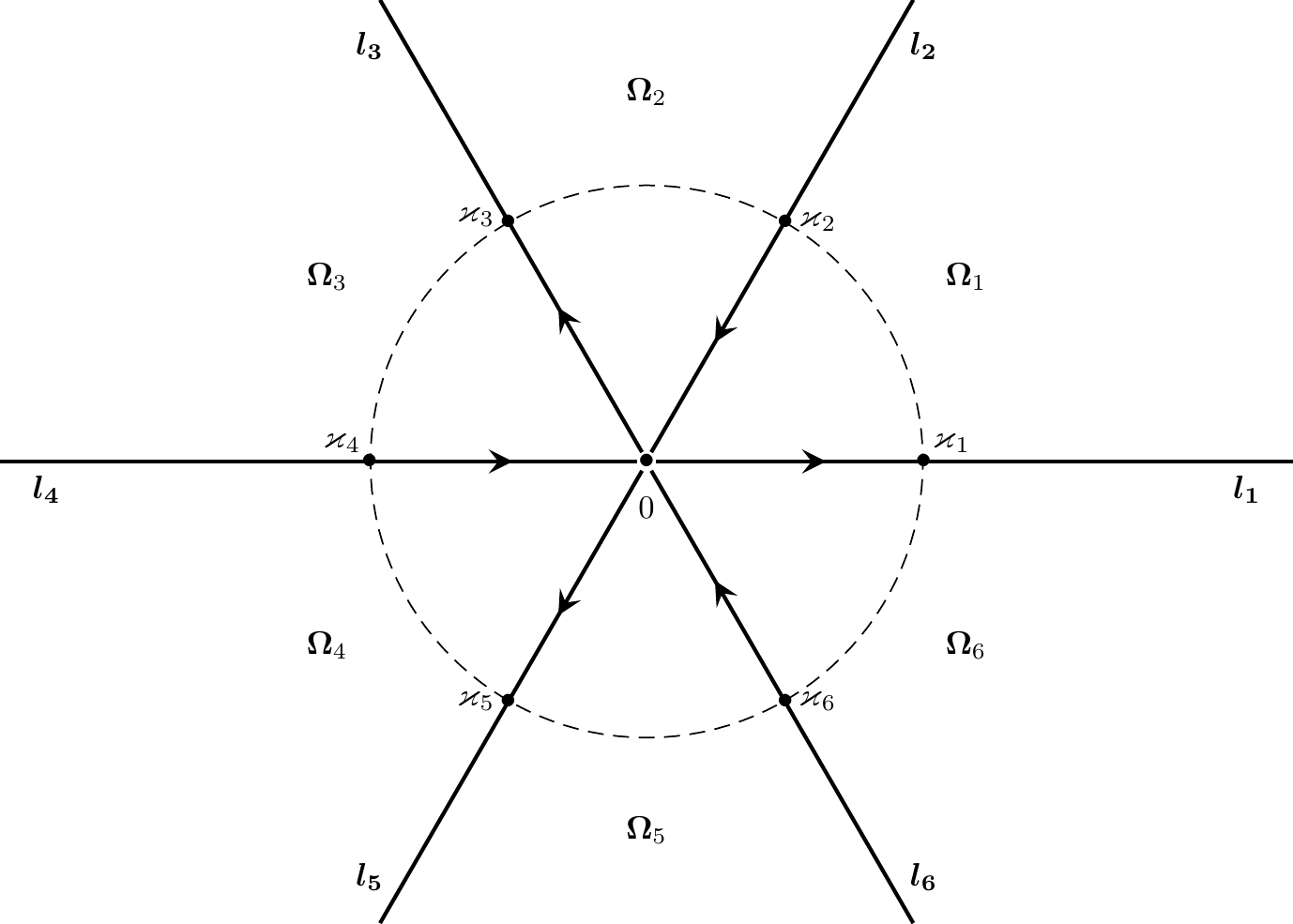}
\caption{Rays $l_{\nu}$, domains $\Omega_{\nu}$ and points $\varkappa_{\nu}$ in the $k$-plane} 
\label{fig:sectors}
\end{figure}

\noindent
The contour $\Sigma=\{k\mid\Re\lambda_j(k)=\Re\lambda_l(k)\text{ for some } j\neq l\}$
consists of six rays
\[
l_\nu=\D{R}_+\ee^{\frac{\ii\pi}{3}(\nu-1)}=\D{R}_+\varkappa_{\nu},\quad\nu=1,\dots,6
\]
dividing the $k$-plane into six sectors
\[
\Omega_\nu=\Bigl\lbrace k\,\Bigm\vert\,\frac{\pi}{3}(\nu-1)<\arg k < \frac{\pi}{3}\nu\Bigr\rbrace,\quad\nu=1,\dots,6.
\]
In order to have a (matrix-valued) solution to \eqref{M-int2} to be analytic in $\D{C}\setminus \Sigma$,
the initial points of integration $\infty_{jl}$ are specified for each matrix entry $(j,l)$, $1\leq j,l\leq 3$ as follows:
\begin{equation}\label{signs}
\infty_{jl}=
\begin{cases}
+\infty,&\text{if }\Re\lambda_j(z)\geq\Re\lambda_l(z),\\ 
-\infty, &\text{if }\Re\lambda_j(z)<\Re\lambda_l(z).
\end{cases}
\end{equation}
That means that we consider the system of Fredholm integral equations, for $1\leq j,l\leq 3$,
\begin{equation} \label{M-int3}
M_{jl}(x,t,z)=
I_{jl}+\int_{\infty_{j,l}}^x\ee^{-\lambda_j(z)\int_x^{\xi}q(\zeta,t)\dd\zeta}\croch{(\hat UM)_{jl}(\xi,t,z)}\ee^{\lambda_l(z)\int_x^{\xi}q(\zeta,t)\dd\zeta}\dd\xi.
\end{equation}

\begin{prop}[analyticity] \label{prop:p1}
Let $M(x,t,k)$ be the solution of the system of equations \eqref{M-int3}, where the limits of integration are chosen according to \eqref{signs}. Then,
\begin{enumerate}[\rm(i)]
\item
$M$ is piecewise meromorphic with respect to $\Sigma$, as function of the spectral parameter $k$. 
\item
$M(x,t,k)\to I$ as $k\to\infty$, where $I$ is the $3\times 3$ identity matrix.
\end{enumerate}
\end{prop}

\begin{proof}
The proof follows the same lines as in \cite{BC}. Notice that in order to have (ii) it is important that the diagonal part of $\hat U$  vanish as $k\to\infty$.
\end{proof}

\begin{prop}[symmetries] \label{p2}
$M(x,t,k)$ satisfies the symmetry relations:
\begin{enumerate}[\rm{(S}1)]
\item 
$\Gamma_1 \overline{M(x,t,\bar k)} \Gamma_1 = M(x,t,k)$ where $\Gamma_1 = 
\left(\begin{smallmatrix}
0 & 1 & 0 \\
1 & 0 & 0 \\
0 & 0 & 1
\end{smallmatrix}\right)$.
\item
$\Gamma_2 \overline{M(x,t,\bar k \omega^2)} \Gamma_2 = M(x,t,k)$ where $\Gamma_2 = 
\left(\begin{smallmatrix}
	0 & 0 & 1 \\
	0 & 1 & 0 \\
	1 & 0 & 0
\end{smallmatrix}\right)$.
\item 
$\Gamma_3 \overline{M(x,t,\bar k \omega)} \Gamma_3 = M(x,t,k)$ where $\Gamma_3 = 
\left(\begin{smallmatrix}
	1 & 0 & 0 \\
	0 & 0 & 1 \\
	0 & 1 & 0
\end{smallmatrix}\right)$.
\item 
$M(x,t,\frac{1}{k})=\overline{M(x,t,\bar k)}$.
\end{enumerate}
\end{prop}

From (S1)-(S3) it follows that the values of  $M$ at $k$ and at $\omega k$ are related by
\[
M( x,t,k \omega) = C^{-1}M(x,t,k)C,\ \text{ where }\ 
C=\left(\begin{smallmatrix}
	0 & 0 & 1 \\
	1 & 0 & 0 \\
	0 & 1 & 0
\end{smallmatrix}\right).
\]

If $\lambda_j(k)=\lambda_l(k)$, $j\neq l$ for certain values of the spectral parameter $k$, then $P$ at these values becomes degenerate (see \eqref{La-P-1}), which in turn leads to singularities for $\hat U$ and, consequently, for $\hat\Phi$ and $M$. These values are $\varkappa_{\nu}=\ee^{\frac{\ii\pi}{3}(\nu-1)}$, $\nu=1,\dots,6$. Taking into account the symmetries described in Proposition \ref{p2} leads to the following

\begin{prop}[singularities] \label{p3}
The limiting values of $M(x,t,k)$ as $k$ approaches one of the points $\varkappa_\nu=\ee^{\frac{\ii\pi}{3}(\nu-1)}$, $\nu=1,\dots,6$ have pole singularities with leading terms of a specific matrix structure. 
\begin{enumerate}[\rm(i)]
\item
As $k\to\varkappa_1=1$,
\begin{equation}\label{singul-1}
M\sim \frac{1}{k-1}
\begin{pmatrix}
	1 & 1 & 1 \\
	-1 & -1 & -1 \\
	0 & 0 & 0
\end{pmatrix}
\begin{pmatrix}
	\alpha & 0 & 0 \\
	0 & \alpha & 0 \\
	0 & 0 & \beta
\end{pmatrix},
\end{equation} 
where $\alpha=-\overline{\alpha}$ and $\beta=-\overline{\beta}$.
\item
As $k\to\varkappa_2=\ee^{\frac{\ii\pi}{3}}$,
\begin{equation}\label{singul-2}
M\sim \frac{1}{k-\ee^{\frac{\ii\pi}{3}}}
\begin{pmatrix}
	0 & 0 & 0 \\
		1 & 1 & 1  \\
		-1 & -1 & -1
\end{pmatrix}
\begin{pmatrix}
	 \tilde\beta & 0 & 0 \\
	   0 & \tilde\alpha & 0 \\
	   0 & 0 & \tilde\alpha 
\end{pmatrix},
\end{equation} 
where $\tilde\alpha=-\omega\overline{\tilde\alpha}$ and $\tilde\beta=-\omega\overline{\tilde\beta}$.
\end{enumerate}
$M$ has similar leading terms at the other polar singularities $\varkappa_3,\dots,\varkappa_6$ in accordance with the symmetry conditions stated in Proposition \ref{p2}.
\end{prop}

\begin{proof}
Indeed, consider, for example, the behavior of $M$ as $k\to\varkappa_1=1$. We have $\lambda_1(1)=\lambda_2(1)=-1/\sqrt{3}$ and  $\lambda_3(1)=2/\sqrt{3}$. Moreover, as $k\to 1$, we have
\begin{align*}
3\lambda_1^2-1 &= -2\sqrt{3}(k-1) + \ord\left((k-1)^2\right), \\
3\lambda_2^2-1 &= 2\sqrt{3}(k-1) + \ord\left((k-1)^2\right), \\
3\lambda_3^2-1 &=3.
\end{align*}
Consequently, as $k\to\varkappa_1=1$,
\[
P^{-1}(k) = \frac{1}{k-1}\begin{pmatrix}
1 & 1 & 1 \\
-1 & -1 & -1 \\
0 & 0 & 0
\end{pmatrix}
\begin{pmatrix}
p_1 & 0 & 0 \\
0 & p_2 & 0 \\
0 & 0 & p_3
\end{pmatrix} + \ord(1)
\]
with some $p_j$'s.

Now recall that $\hat\Phi = P^{-1}\tilde\Phi$, where, on one hand, $\tilde\Phi$ satisfies a differential equation whose coefficients are regular at $k=1$, and on the other hand, $\tilde\Phi$ satisfies the boundary condition (see \eqref{M-int2} with signs chosen according to \eqref{signs})
\[
\tilde\Phi\sim P(k)\ee^{y(x,t)\Lambda(k)+t A(k)},\quad x\to +\infty
\]
for all $k$ such that $\lambda_j(k)\neq\lambda_l(k)$ for all $j\neq l$. It follows that when $k$ approaches $\varkappa_1=1$ from the either side of $l_1$, the first two columns of $\tilde\Phi$ coincide. Consequently, for the limiting values of $\hat\Phi$ (and thus of $M$) we obtain \eqref{singul-1} while the property  $\alpha=-\overline{\alpha}$ and $\beta=-\overline{\beta}$ follows from the symmetries (S1) and (S4) of Proposition \ref{p2}. 

Similarly for the other points $\varkappa_{\nu}=\ee^{\frac{\ii\pi}{3}(\nu-1)}$, $\nu=2,\dots,6$.
\end{proof}

\begin{rem}
While the set of singularities of $M(x,t,k)$ in the open set 
\[
\Omega=\D{C}\setminus\Sigma=\Omega_1\cup\dots\cup\Omega_6
\]
can be empty (for instance, this happens for all sufficiently small ``potentials'' $u(x,t)$), the singularities described in Proposition \ref{p3} are generic. This should be compared with \cite{CIL10}, where the solutions of a system of integral equations similar to \eqref{M-int2} are combined into a Riemann--Hilbert problem under the assumption that they have no singularities.
\end{rem}

\section{Riemann--Hilbert problems}
\subsection{Jump conditions}

\subsubsection*{Matrix RH problem}

For $k$ on the boundary between two adjacent domains $\Omega_\nu$, the limiting values of $M$, being the solutions of the system of differential equations \eqref{M} must be related by a matrix independent of $x$ and $t$. Supplying the rays $l_\nu$ with the orientation, see Figure~\ref{fig:sectors}, we can write for the limiting values of $M$:
\begin{equation}\label{scat}
M_+(x,t,k)=M_-(x,t,k)\ee^{Q(x,t,k)}S_0(k)\ee^{-Q(x,t,k)},\quad k\in l_1\cup\dots\cup l_6.
\end{equation}
Considering \eqref{scat} at $t=0$ we see that $S_0(k)$ is, in fact, determined by $u(x,0)$, i.e., by the initial data for the Cauchy problem \eqref{DP-om}--\eqref{IC}, via the solutions $M(x,0,k)$ of the system \eqref{M-int3} whose coefficients are determined by $u(x,0)$. Thus the relation \eqref{scat} can be considered as a ``pre-Riemann--Hilbert problem'' associated with \eqref{DP-om}: the data are $S_0(k)$, and we seek for a piecewise meromorphic function $M$ satisfying \eqref{scat} for all $x$ and $t$, in the hope that one can further extract the solution $u(x,t)$ to \eqref{DP-om} from $M(x,t,k)$. 

The matrix $S_0(k)$ has a particular matrix structure, cf.~\cite{CIL10}. Indeed, the integral equations \eqref{M-int3} allow studying the limiting values of $M$ as $x\to\pm\infty$. 

\subsubsection*{$S_0(k)$ for $k\in l_\nu$}

Set $t=0$ and consider, for example, the limiting values of $M_\pm(x,0,k)$ for $k\in l_1=\D{R}_+$. For such $k$, $\Re\lambda_1(k)=\Re\lambda_2(k)<\Re\lambda_3(k)$, and the triangular structure of integration in \eqref{M-int3} (see \eqref{signs}) implies that, as $x\to+\infty$,
\begin{equation}\label{lims-p}
M_+(x,0,k)=\begin{pmatrix}
1 & r_+(k){E}(x,k) & 0\\
0 & 1 & 0\\
0 & 0 & 1
\end{pmatrix} 
+\osmall(1)
\end{equation}
and 
\begin{equation}\label{lims-m}
M_-(x,0,k)=\begin{pmatrix}
1 & 0 & 0\\
r_-(k){E}^{-1}(x,k) & 1 & 0\\
0 & 0 & 1
\end{pmatrix}
+\osmall(1),
\end{equation}
where ${E}(x,k)=\ee^{Q_{11}(x,0,k)-Q_{22}(x,0,k)}=\ee^{y(x,0)(\lambda_1(k)-\lambda_2(k))}$, whereas $M_+$ and $M_-$ are bounded for all $x$ (particularly, as $x\to\infty$). The symmetry (S1) from Proposition \ref{p2} implies that $r_-(k)=\overline{r_+(k)}$.

Letting $x\to+\infty$ in the r.h.s.\ of 
\[
S_0(k)=\ee^{-Q(x,0,k)}M_-^{-1}(x,0,k)M_+(x,0,k)\ee^{Q(x,0,k)}
\]
we get that $(S_0)_{31}(k)=(S_0)_{32}(k)\equiv 0$ and $(S_0)_{11}(k)=(S_0)_{33}(k)\equiv 0$. On the other hand, letting $x\to-\infty$ yields $(S_0)_{13}(k)=(S_0)_{23}(k)\equiv 0$.
Thus for $k\in l_1$
\begin{equation}\label{S-SS-1}
S_0(k)= 
\begin{pmatrix}
1 & 0 & 0\\
-r(k) & 1 & 0\\
0 & 0 & 1
\end{pmatrix}
\begin{pmatrix}
1 & \overline{r(k)} & 0\\
 0 & 1 & 0\\
0 & 0 & 1
\end{pmatrix},
\end{equation}
where $r(k):=r_-(k)=\overline{r_+(k)}$, whereas the symmetry (S4) implies that $r(1/k)=\overline{r(k)}$. 

Similarly, for $k\in l_4=\D{R}_-$, the matrix $S_0(k)$ has the same form:
\begin{equation}\label{S-SS-4}
S_0(k)=\left(G_4^-(k)\right)^{-1}G_4^+(k)= 
\begin{pmatrix}
1 & 0 & 0\\
-r(k) & 1 & 0\\
0 & 0 & 1
\end{pmatrix}
\begin{pmatrix}
1 & \overline{r(k)} & 0\\
 0 & 1 & 0\\
0 & 0 & 1
\end{pmatrix},
\end{equation}
whereas the construction of $S_0(k)$ for $k\in l_{\nu}$ with $\nu\neq 1,4$ follows from the symmetries of Proposition \ref{p2}. 

Thus, similarly to the cases of, say, the KdV equation or the Camassa--Holm equation, the jump matrix on the whole contour is determined by a scalar function --- the reflection coefficient $r(k)$ given for $k\in\D{R}$. In turn, as it follows from \eqref{lims-p} and \eqref{M-int3} for $t=0$, the reflection coefficient is determined by the initial condition $u_0(x)$ via the solution of \eqref{M-int3}, where $m(x,0)\equiv u(x,0)-u_{xx}(x,0)$ in the construction of $\hat U(x,0,z)$ is replaced by $m_0(x)\equiv u_0(x)-u_0''(x)$.

\subsubsection*{New matrix RH problem (after change of parameter)}
The dependence  of the matrix $\ee^{Q}S_0\ee^{-Q}$ relating $M_+$ and $M_-$ in \eqref{scat} on the parameters $x$ and $t$ suggests introducing the parameter $y=y(x,t)$ given by \eqref{y} and thus rewriting \eqref{scat} as 
\begin{subequations}
\begin{align}\label{RH-y}
\hat M_+(y,t,k)
&=\hat M_-(y,t,k)\ee^{y\Lambda(k)+tA(k)}S_0(k)\ee^{-y\Lambda(k)-tA(k)}\notag\\
&=\hat M_-(y,t,k)S(y,t,k)
\shortintertext{with}
S(y,t,k)&=\ee^{y\Lambda(k)+tA(k)}S_0(k)\ee^{-y\Lambda(k)-tA(k)},
\end{align}
\end{subequations}
so that 
\[
M(x,t,k)=\hat M(y(x,t),t,k),
\]
which provides explicit dependence of the relating matrix on the parameters $(y,t)$. 

\subsubsection*{Vector RH problem}
On the other hand, the particular matrix structure of the singularities at $k=\varkappa_{\nu}\equiv\ee^{\frac{\ii\pi}{3}(\nu-1)}$ suggests introducing a row vector-valued $1\times 3$ Riemann--Hilbert problem 
\begin{subequations}  \label{RH}
\begin{align}  \label{RH-y-v}
\mu_+(y,t,k)
&=\mu_-(y,t,k)S(y,t,k)
\shortintertext{where as above}
S(y,t,k)&=\ee^{y\Lambda(k)+tA(k)}S_0(k)\ee^{-y\Lambda(k)-tA(k)},\notag
\end{align}
having the same jump conditions across $\Sigma$ as in the matrix case, but with the following normalization condition at $k=\infty$:
\begin{equation}  \label{vec-norm}
\mu(y,t,k)=\begin{pmatrix}1&1&1\end{pmatrix}+\osmall(1)\quad\text{as }k\to\infty.
\end{equation}
The transition from the $3\times 3$ matrix RH problem to the $1\times 3$ row vector RH problem can be viewed as the multiplication of the former by the constant row vector $(1\ \ 1\ \ 1)$ from the left, which suppresses the singularities at $k=\varkappa_{\nu}\equiv\ee^{\frac{\ii\pi}{3}(\nu-1)}$.

\subsection{Residue conditions}

In order to complete the formulation of the RH problem, one has to complete the jump condition \eqref{RH-y-v} and the normalization condition \eqref{vec-norm} with residue conditions at the possible poles of $\mu(y,t,k)$ or $\hat M(y,t,k)$ in $\D{C}\setminus\Sigma$ where $\Sigma=\cup_{\nu=1}^6l_{\nu}$. The following statement holds true (see \cite{BC}).

\emph{Generically, there are at most a finite number of poles lying in $\D{C}\setminus\Sigma$, each of them being simple, with residue conditions of a special matrix form: distinct columns of $M$ (distinct entries of $\mu$) have distinct poles, and if $k_n$ is a pole, then
\begin{equation}\label{res}
\Res_{k=k_n}\mu(y,t,k)=\lim_{k\to k_n}\mu(y,t,k)\ee^{y\Lambda(k)+tA(k)}v_n\ee^{-y\Lambda(k)-tA(k)},
\end{equation}
where the $3\times 3$ matrix $v_n$ has only one non-zero entry at a position depending on the sector of $\D{C}\setminus\Sigma$ containing $k_n$.}
\end{subequations}

For example, if $k_n\in\Omega_1$, a non-zero entry of $v_n$ can be either $(v_n)_{12}$ or $(v_n)_{23}$. Then the positions (as well as the values) of the non-zero entries of the poles in other sectors of $\D{C}\setminus\Sigma$ are determined by the symmetries (S1)-(S3). 

Similarly to $S_0(k)$, the residue conditions are determined by the initial values $u_0(x)$ (putting $t=0$ in \eqref{res}).

A closer look at the $(y,t)$-dependence of the exponential in the residue condition \eqref{res} reveals the following. If $w_n:=\ee^{y\Lambda(k_n)+tA(k_n)}v_n\ee^{-y\Lambda(k_n)-tA(k_n)}$, and if $(v_n)_{jl}$ is the non-zero-entry of $v_n$, then
\[
w_n=\ee^{y(\lambda_j(k_n)-\lambda_l(k_n))+t(A_j(k_n)-A_l(k_n))}\,v_n.
\]
Moreover, the exponential factor has the form (see \eqref{la-a})
\begin{align}\label{res-yt}
\ee^{y(\lambda_j(k_n)-\lambda_l(k_n))+t(A_j(k_n)-A_l(k_n))}
&=\ee^{(\lambda_j(k_n)-\lambda_l(k_n))\bigl(y-\frac{t}{\lambda_j(k_n)\lambda_l(k_n)}\bigr)}\notag\\
&=\ee^{g\bigl(y-\frac{3t}{1-v^2}\bigr)},
\end{align}
where
\[
g=\lambda_j(k_n)-\lambda_l(k_n).
\]
The last equality in \eqref{res-yt} follows from the fact (see the definition \eqref{la-eqn} of $\lambda_j$) that $3\lambda_j\lambda_l=1-(\lambda_j-\lambda_l)^2$ for all $j\neq l$. 

A visible analogy of \eqref{res-yt} with the structure of the residue conditions in the case of the Camassa--Holm equation (see (3-8) in \cite{BS08}) suggests making the conjecture that the actual positions of the poles associated with a global solution of the initial-value problem \eqref{DP-om} with smooth, decaying, real-valued initial data $u_0(x)$ are those for which the associated value of $g$ is real and, moreover, $0<\abs{g}<1$. Notice that for the Camassa--Holm equation, the discrete spectrum of the linear operator (from the associated Lax pair) generated by smooth, decaying, real-valued initial data indeed gives rise to residue conditions of the form \eqref{res-yt} with parameters satisfying the requirements above, see \cite{C01,BKST}.

Due to the symmetry relations, it is enough to consider the location of the poles in one sector of $\D{C}\setminus\Sigma$, say, in $\Omega_1$. 

The exponential factor for $(v_n)_{12}$ contains
\begin{align}\label{la1-2}
\lambda_1(k)-\lambda_2(k) 
&= \frac{1}{\sqrt{3}}\left(\omega k+\frac{1}{\omega k}-\omega^2k-\frac{1}{\omega^2 k}\right)\notag\\
&=\frac{\omega-\omega^2}{\sqrt{3}}\left(k-\frac{1}{k}\right)=\ii\left(k-\frac{1}{k}\right).
\end{align}
It follows that the l.h.s.\ of \eqref{la1-2} is real either for $\{k\mid\Re k =0\}$ or for $\{k\mid|k|=1\}$. The imaginary axis does not intersect with $\Omega_1$ whereas for $k=\ee^{\ii\varphi}$ with $0<\varphi<\pi/6$, the associated $g=\lambda_1-\lambda_2$ satisfies the inequalities $0<\abs{g}<1$.

The exponential factor for $(v_n)_{23}$ contains
\begin{equation}\label{la2-3}
\lambda_2(k)-\lambda_3(k) = \frac{1}{\sqrt{3}}
\left(\omega^2 k +\frac{1}{\omega^2 k}- k -\frac{1}{ k}\right) = 
\ii \left(\omega k-\frac{1}{\omega k}\right).
\end{equation}
Thus the admissible arc for the location of the poles in $\Omega_1$ for this entry $(23)$ is $\{k\mid k=\ee^{\ii\varphi},\ \pi/6<\varphi<\pi/3\}$. Notice that the symmetry relations provide that the associated residue conditions in $\Omega_3$ are those for the $(12)$ entry. Similarly for the other sectors $\Omega_{\nu}$.

In analogy with the Camassa--Holm equation (cf.~\cite{BS-D}), the poles with the residue conditions described above can be associated with the soliton long-time behavior of the solution of the Cauchy problem \eqref{DP-om}-\eqref{IC}, the velocity of all solitons being greater than $3$ (since $|g|<1$ --- the velocity $c$ of a soliton being related to $g$ as $c=3/(1-g^2)$, see \eqref{res-yt}). On the other hand, the form of the residue conditions with poles at $\{k\mid k\in(\ee^{\frac{\ii\pi}{6}+\frac{\ii\pi l}{3}})\D{R}_+,\,l=0,\dots,5\}$ suggests that the solution of the associated RH problem gives rise to ``loop solitons'' (a direct construction of the loop solitons is given in \cite{M06}): in this case,
\[
g=\ii\left(\ii\rho-\frac{1}{\ii\rho}\right)=-\left(\rho+\frac{1}{\rho}\right),\quad\rho\in\D{R}
\]
and thus $|g|>2$ (see (30a) in \cite{M06} with $\kappa k_j>2$); so they move, as opposite to the smooth solitons, in the negative  direction. These solutions are not classical ones: in the $(y,t)$ scale, they are given, similarly to the solitons, in an univalent way, but the transition to the original $(x,t)$ scale makes them multivalued.

A detailed account for the soliton solutions in the framework of the RH approach will be given elsewhere.

\subsection{Solution of the Cauchy problem in terms of the solution of the RHP}

\subsubsection*{Assumption}
In what follows we assume that the RH problem consisting in finding a piecewise, vector-valued function $\mu$ satisfying the jump condition \eqref{RH-y-v}, the normalization condition \eqref{vec-norm}, the symmetry conditions 
\[
\mu(k)=\mu(k\omega)\cdot C=\overline{\mu(\bar k)}\cdot\Gamma_1=\overline{\mu(\bar k^{-1})}
\]
(cf.~Proposition~\ref{p2}), and the appropriate residue conditions has a unique solution. 

In order to obtain $u(x,t)$ from its solution $\mu(y,t,k)$, it turns out that it is convenient to evaluate $\mu(y,t,k)$ at distinguished points in the $k$-plane, more precisely, at the points 
\[
\kappa_{\nu}=\ee^{\frac{\ii\pi}{6}}+\ee^{\frac{\ii\pi}{3}(\nu-1)},\quad\nu=1,\dots,6
\] 
characterized by the property $z(\kappa_{\nu})=0$.

\subsubsection*{Last Lax pair}
Coming back to the system \eqref{Lax-vec}, let us introduce $\tilde\Phi^{(0)}=\tilde\Phi^{(0)}(x,t,z)$ by
\[
\tilde\Phi^{(0)}=P^{-1}\Phi.
\] 
This reduces \eqref{Lax-vec} to 
\begin{subequations}  \label{Lax-vec-0}
\begin{align}\label{Lax-0}
\tilde\Phi^{(0)}_x -\Lambda(z)\tilde\Phi^{(0)}&= \tilde U^{(0)}\tilde\Phi^{(0)},\\
\tilde\Phi^{(0)}_t - A(z)\tilde\Phi^{(0)} &=  \tilde V^{(0)}\tilde\Phi^{(0)},
\end{align}
where 
\begin{align}\label{U-0}
\tilde U^{(0)}(x,t,z)&=z^3 m(x,t) 
\begin{pmatrix}
\frac{1}{3\lambda_1^2(z)-1} & 0 & 0 \\
0 & \frac{1}{3\lambda_2^2(z)-1} & 0 \\
0 & 0 &  \frac{1}{3\lambda_3^2(z)-1}
\end{pmatrix}
\begin{pmatrix}
1 & 1 & 1 \\
1 & 1 & 1 \\
1 & 1 & 1 
\end{pmatrix},\\
\label{V-0}
\tilde V^{(0)}(x,t,z)&=P^{-1}(z)
\begin{pmatrix}
u_x & -u & 0 \\
u & 0 & -u \\
u_x-z^3u(m+1) & 0 &  -u_x
\end{pmatrix}P(z).
\end{align}
\end{subequations}
Now notice that $ \tilde U^{(0)}(x,t,z)\big|_{z=0}\equiv 0$. Therefore, introducing $M^{(0)}=M^{(0)}(x,t,z)$ by 
\[
M^{(0)}=\tilde\Phi^{(0)}\ee^{-x\Lambda -tA},
\]
and determining $M^{(0)}$ as the solution of a system of integral equations similar to the system \eqref{M-int3} determining $M$, we have:
\[
M^{(0)}(x,t,z)\big|_{z=0}\equiv I.
\]
On the other hand, since $M^{(0)}$ and $M$ are solutions of differential equations coming from the same system of differential equations \eqref{Lax-vec}, and since they have the same limit as $x\to +\infty$ for $k\not\in\Sigma$:
\[
M,\,M^{(0)}\xrightarrow[x\to+\infty]{}I,
\]
it follows that they are related by
\begin{equation}\label{MM-rel}
M(x,t,k)=P^{-1}(k)D^{-1}(x,t)P(k) M^{(0)}(x,t,k)\ee^{(x-y(x,t))\Lambda(k)}.
\end{equation}
Particularly, at $k=\kappa_1\equiv\ee^{\frac{\ii\pi}{6}}$ we have 
\begin{subequations}\label{k0}
\begin{align}
\Lambda(\ee^{\frac{\ii\pi}{6}})&= 
\begin{pmatrix}
-1 & 0 & 0 \\
0 & 0 & 0 \\
0 & 0 & 1
\end{pmatrix},\\
P^{-1}(k)D^{-1}(x,t)P(k)\big|_{k=\ee^{\frac{\ii\pi}{6}}}&= 
\frac{1}{2}\begin{pmatrix}
1+\frac{1}{q} & 0 & -1+\frac{1}{q} \\
2\bigl(q-\frac{1}{q}\bigr) & 2q &  2\bigl(q-\frac{1}{q}\bigr) \\
-1+\frac{1}{q} & 0 & 1+\frac{1}{q}
\end{pmatrix}.
\end{align}
\end{subequations}
Now observe that 
\[
\begin{pmatrix}1&1&1\end{pmatrix}P^{-1}(k)D^{-1}(x,t)P(k)\big|_{k=\ee^{\frac{\ii\pi}{6}}}=q(x,t)\begin{pmatrix}1&1&1\end{pmatrix}.
\]
Combined with \eqref{MM-rel}, this implies that the row vector solution $\mu(y,t,k)$ evaluated at $k=\kappa_1\equiv\ee^{\frac{\ii\pi}{6}}$ takes the value
\begin{align}\label{mu-0}
\mu(y,t,\ee^{\frac{\ii\pi}{6}}) 
&=q(x,t)\begin{pmatrix}1&1&1\end{pmatrix}
\begin{pmatrix}
\ee^{-\int_x^\infty (q(\xi,t)-1)\dd\xi} & 0 & 0 \\
0 & 1 & 0 \\
0 & 0 & \ee^{\int_x^\infty(q(\xi,t)-1)\dd\xi}
\end{pmatrix}\notag\\
&=\begin{pmatrix}
q(x,t)\ee^{-\int_x^\infty (q(\xi,t)-1)\dd\xi}&q(x,t)&q(x,t)
\ee^{\int_x^\infty (q(\xi,t)-1)\dd\xi}
\end{pmatrix}.
\end{align}
Taking into account that in terms of functions of the variables $(x,t)$ we have $u(y,t)=\frac{\partial x}{\partial t}(y,t)$ (this follows from \eqref{DP-b} and from the definition \eqref{y} of the new variable $y$), the relation \eqref{mu-0} provides a parametric representation of $u$.

\begin{thm}\label{u-repr}
Let $\mu\equiv\mu(y,t,k)=\begin{pmatrix}\mu_1&\mu_2&\mu_3\end{pmatrix}$ be the solution of the Riemann--Hilbert problem \eqref{RH-y-v}--\eqref{res}, where $S(k)$ is the scattering matrix and $\{v_n\}$ are the residues associated with the initial data $u_0(x)$. 

Then the solution $u(x,t)$ of the Cauchy problem \eqref{DP-om}--\eqref{IC} for the Degasperis--Procesi equation can be expressed in terms of $\mu(y,t,k)$, evaluated at $k=\ee^{\frac{\ii\pi}{6}}$, in parametric form:
\begin{align}\label{u-y}
u(y,t) &= \frac{\partial}{\partial t}\log \frac{\mu_{j+1}}{\mu_j}(y,t,\ee^{\frac{\ii\pi}{6}}), \\ \notag
x(y,t) &= y + \log \frac{\mu_{j+1}}{\mu_j}(y,t,\ee^{\frac{\ii\pi}{6}}),\qquad j=1\ \text{or}\ 2.
\end{align}
\end{thm}

Notice that the structure of the parametric representation \eqref{u-y} is similar to that in the case of the Camassa--Holm equation. Moreover, this structure appears also in formulae for pure multisoliton solutions given in \cite{M05}.

\section{Long-time asymptotics}

\subsection{Qualitative results}   \label{asy.results}

The analysis of the long-time behavior of the solution of the IVP is based on the analysis of the large-$t$ behavior of the solution of the associated RH problem. The latter can be done in the framework of the nonlinear steepest descent method, whose key ingredient is the deformation of the original RH problem in accordance with the ``signature table'' for the phase functions involved in the jump matrix. 

The analysis presented in the previous section shows that the structure of the jump matrix $S$, which is $3\times 3$, is essentially $2\times 2$: for each straight line of the contour $\Sigma$, a non-trivial block of $S(k)$ is  $2\times 2$, see \eqref{S-SS-1}, \eqref{S-SS-4} (under an appropriate change of basis), with only one exponential involved for each part, of the form \eqref{res-yt}. Now observe that this exponential is essentially the same as in the case of the Camassa--Holm equation: see \cite[Eq.~(3.1)]{BS-D}); only the constant $2$ in the numerator is replaced by $3$. Consequently, the deformations of the each part of $\Sigma$ are performed in the same way as the deformation of the real line in the case of the Camassa--Holm equation, leading to the similar asymptotic behavior of the solution, see \cite{BS-D, BKST}:
\begin{figure}[ht]
\centering\includegraphics[scale=1]{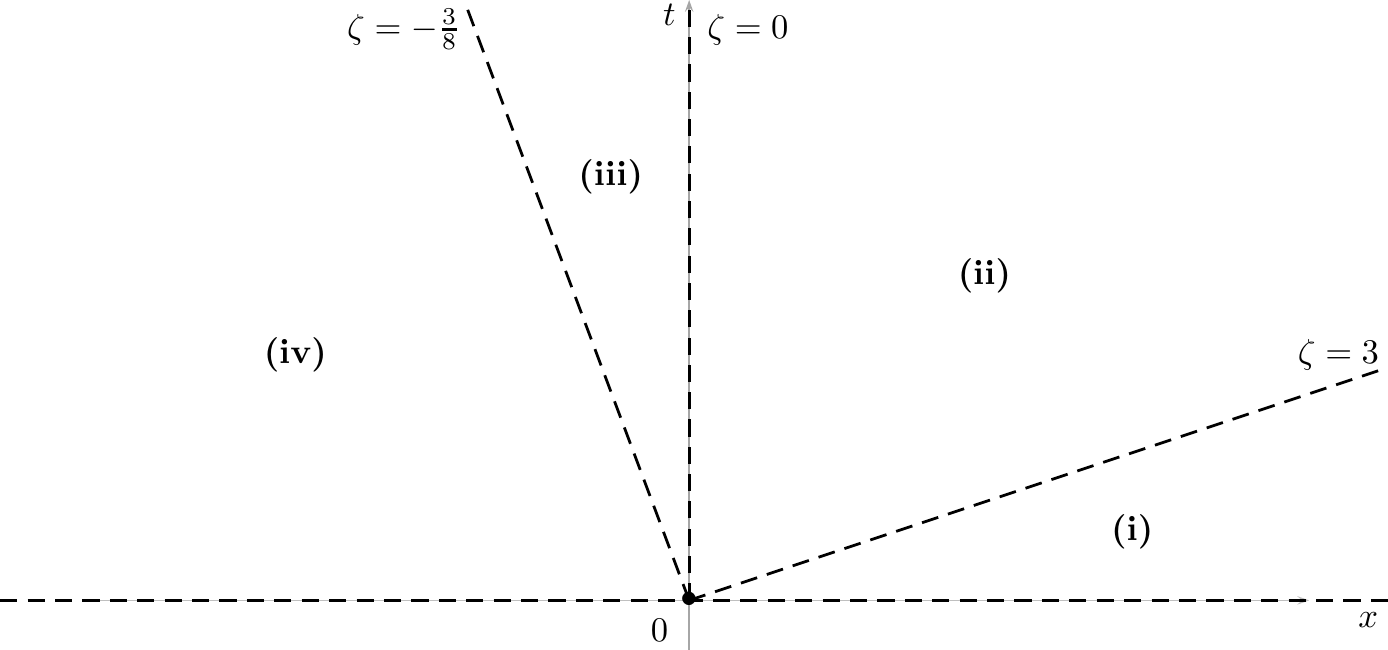}
\caption{Sectors in the $(x,t)$-half-plane with different asymptotics, $\zeta=x/t$} 
\label{fig:asymptotic.sectors}
\end{figure}
\begin{enumerate}[(i)]
\item
In the sector $\frac{x}{t}>3$, the long-time  asymptotics is dominated by the solitons; in the case when there are no solitons, the asymptotics in this sector is fast decaying.
\item
For $0<\frac{x}{t}<3$, the asymptotics has the form of slowly decaying
(as $t^{-1/2}$) modulated oscillations.
\item
For $-\frac{3}{8}<\frac{x}{t}<0$, the asymptotics has the form of a sum of two slowly decaying modulated oscillations.
\item
For $\frac{x}{t}<-\frac{3}{8}$, the asymptotics is fast decaying.
\end{enumerate}
Besides, there are transition zones between the sectors, where the (slowly decaying) asymptotics is given in terms of the dedicated solutions of the Painlev\'e II equation.

To illustrate the asymptotics, here we present the asymptotic formula for the first sector of decaying oscillations $0<\frac{x}{t}<3$, see Section~\ref{asy.sector.two}. For other asymptotic regions, see Section~\ref{asy.other.sectors}.

\subsection{Asymptotics in sector (ii)}   \label{asy.sector.two}

In the next theorem we give the asymptotic formula for the sector $0<\zeta=\frac{x}{t}<3$, assuming for simplicity that there are no discrete eigenvalues. 

\begin{thm}[asymptotics for $0<\zeta<3$] \label{t4}
Let $M(x,t,k)$ be the solution of \eqref{M-int3} associated with the solution of the initial value problem \eqref{DP-om}-\eqref{IC}. Assume $M$ has no singularities other than those described in Proposition~\ref{p3}. 

Then for $\varepsilon<\frac{x}{t}<3-\varepsilon$ for any $\varepsilon>0$, the solution $u(x,t)$ of \eqref{DP-om}-\eqref{IC} behaves, as $t\to\infty$, as follows:
\begin{equation}   \label{asymptotics}
u(x,t)=\frac{c_1}{\sqrt t}\sin(c_2t + c_3\log t+c_4)(1+\osmall(1))
\end{equation}
where $c_1,\dots,c_4$ are functions of $\zeta=\frac{x}{t}$, which are determined by $u_0(x)$ in terms of the associated reflection coefficient $r(k)$, see \eqref{eq:tc14}-\eqref{eq:c14} below.
\end{thm}

\begin{proof}[Sketch of proof] 
Consider the ``signature tables'' for the exponentials in the jump matrix 
\[
S(y,t,k):=\ee^{y\Lambda(k)+tA(k)}S_0(k)\ee^{-y\Lambda(k)-tA(k)}
\]
in \eqref{RH-y-v}. For $k\in\D{R}$ (see \eqref{S-SS-1} and \eqref{S-SS-4}), the exponentials are involved in the $(12)$ and $(21)$ entries. Introducing 
\[
\tilde k(k)=\tfrac{1}{2}\bigl(\tfrac{1}{k}-k\bigr),
\]
and using \eqref{la1-2} one can write $g\equiv\lambda_1(k)-\lambda_2(k)=-2\ii\tilde k$. Then, using \eqref{res-yt}, one gets:
\[
S_{12}(y,t,k)=\bar r(k)\ee^{(\lambda_1(k)-\lambda_2(k))\bigl(y-\frac{t}{\lambda_1(k)\lambda_2(k)}\bigr)}\equiv\bar r(k)\ee^{-2\ii\tilde t\Theta(\xi,\tilde k(k))},
\]
where $\xi:=\frac{2y}{3t}$, $\tilde t:=\frac{3}{2}t$, and thus
\[
S(k)=\begin{pmatrix}
1&\bar r(k)\ee^{-2\ii\tilde t\Theta(\xi,\tilde k(k))}&0\\
-r(k)\ee^{2\ii\tilde t\Theta(\xi,\tilde k(k))}&1-\abs{r(k}^2&0\\
0&0&1
\end{pmatrix},\quad k\in\D{R},
\]
where $\Theta(\xi,\tilde k)$ is exactly the same as for the Camassa--Holm equation \cite{BKST,BS-D}:
\[
\Theta(\xi,\tilde k)=\xi\tilde k-\frac{2\tilde k}{1+4\tilde k^2}\,.
\]
Therefore, the distribution of signs of $\Im\Theta(\xi,\tilde k(k))$ in the $k$-plane can be obtained applying the inverse map $\tilde k\leadsto k$ to the ``signature tables'' for the Camassa--Holm equation \cite{BKST,BS-D}. In the $\tilde k$-plane, denoting $\tilde k=\tilde k_1+\ii\tilde k_2$, one has
\[
\Im\Theta(\xi,\tilde k)=\tilde k_2\times\hat\Theta(\xi,\tilde k_1,\tilde k_2),
\]
where
\[
\hat\Theta(\xi,\tilde k_1,\tilde k_2)=\xi-\frac{2\croch{1-4(\tilde k_1^2+\tilde k_2^2)}}{\croch{1+4(\tilde k_1^2-\tilde k_2^2)}^2+64\tilde k_1^2\tilde k_2^2}\,.
\]
Thus, in the $\tilde k$-plane the real critical points $\tilde k_1\in\D{R}$, where $\hat\Theta(\xi,\tilde k_1,0)=0$, are determined by the equation
\begin{equation}  \label{eq:tildek.critical}
\xi=\frac{2(1-4\tilde k_1^2)}{(1+4\tilde k_1^2)^2}\,,
\end{equation}
or, equivalently, in terms of $\varpi=4\tilde k_1^2+1\geq 1$, by
\begin{equation}   \label{eq:varpi}
\xi\varpi^2+2\varpi-4=0. 
\end{equation}
In the range $0<\xi<2$ this equation has exactly one solution $\geq 1$, which gives in the $\tilde k$-plane two real critical points $\pm\kappa_0$:
\begin{equation}   \label{kappa-0}
\kappa_0(\xi)=\left(\frac{\sqrt{1+4\xi}-1-\xi}{4\xi}\right)^{\frac{1}{2}}.
\end{equation}
Applying the inverse map $\tilde k\leadsto k$ to obtain the corresponding critical points in the $k$-plane we get four real critical points $\pm p_0$, $\pm\frac{1}{p_0}$: 
\begin{equation}   \label{p-0}
p_0=p_0(\xi)=-\kappa_0+\sqrt{\kappa_0^2+1}\,,
\end{equation}
so that
\[
\kappa_0=\frac{1}{2}\Bigl(\frac{1}{p_0}-p_0\Bigr)=\tilde k(p_0)=-\tilde k\bigl(\tfrac{1}{p_0}\bigr).
\]
Thus, in the range $0<\xi<2$ the distribution of signs of $\Im\Theta(\xi,\tilde k(k))$ in the $k$-plane has the form seen in Figure~\ref{fig:signs.real.axis}.
\begin{figure}[ht]
\centering\includegraphics[scale=1]{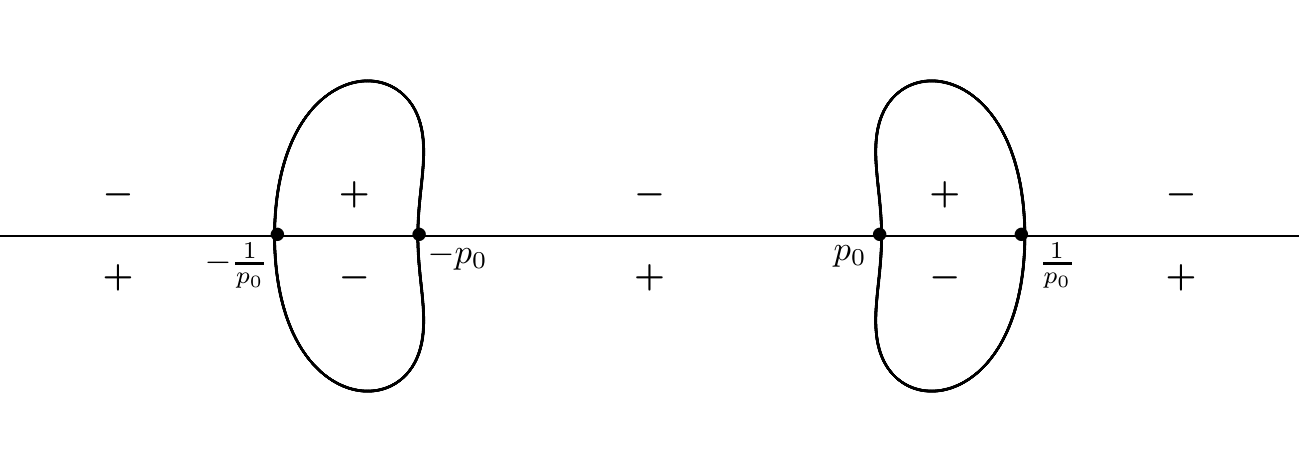}
\caption{Signs of $\Im\Theta(\xi,\tilde k(k))$ in the $k$-plane near $\D{R}$ for $0<\xi=\frac{2y}{3t}<2$} 
\label{fig:signs.real.axis}
\end{figure}

Note that \eqref{y} and estimates below imply that $\frac{y}{t}\sim\frac{x}{t}$ as $t\to\infty$. Thus the region $0<\xi=\frac{2y}{3t}<2$ is asymptotically equivalent to the sector $0<\zeta=\frac{x}{t}<3$. In the considered range of $\xi$, one has $\kappa_0\in(0,\frac{1}{2})$ and thus $p_0\in\bigl(\frac{\sqrt{5}-1}{2},1\bigr)$.

The symmetries of $S$ (see Proposition \ref{p2}) yield
\begin{alignat*}{2}
S(y,t,k)&=
\begin{pmatrix}
1-\abs{r(\omega^2k}^2&0&-r(\omega^2k)\ee^{2\ii\tilde t\Theta(\xi,\tilde k(\omega^2k))}\\
0&1&0\\
\bar r(\omega^2k)\ee^{-2\ii\tilde t\Theta(\xi,\tilde k(\omega^2k))}&0&1
\end{pmatrix},&\quad&k\in\omega\,\D{R}
\shortintertext{and}
S(y,t,k)&=
\begin{pmatrix}
1&0&0\\
0&1&\bar r(\omega k)\ee^{-2\ii\tilde t\Theta(\xi,\tilde k(\omega k))}\\
0&-r(\omega k)\ee^{2\ii\tilde t\Theta(\xi,\tilde k(\omega k))}&1-\abs{r(\omega k}^2
\end{pmatrix},&&k\in\omega^2\,\D{R}.
\end{alignat*}
It follows that the distribution of signs of $\Im\Theta$ for $k$ near each part (line) of $\Sigma$ is as shown in Figure~\ref{fig:signs.parts.sigma}.
\begin{figure}[ht]
\centering\includegraphics[scale=.8]{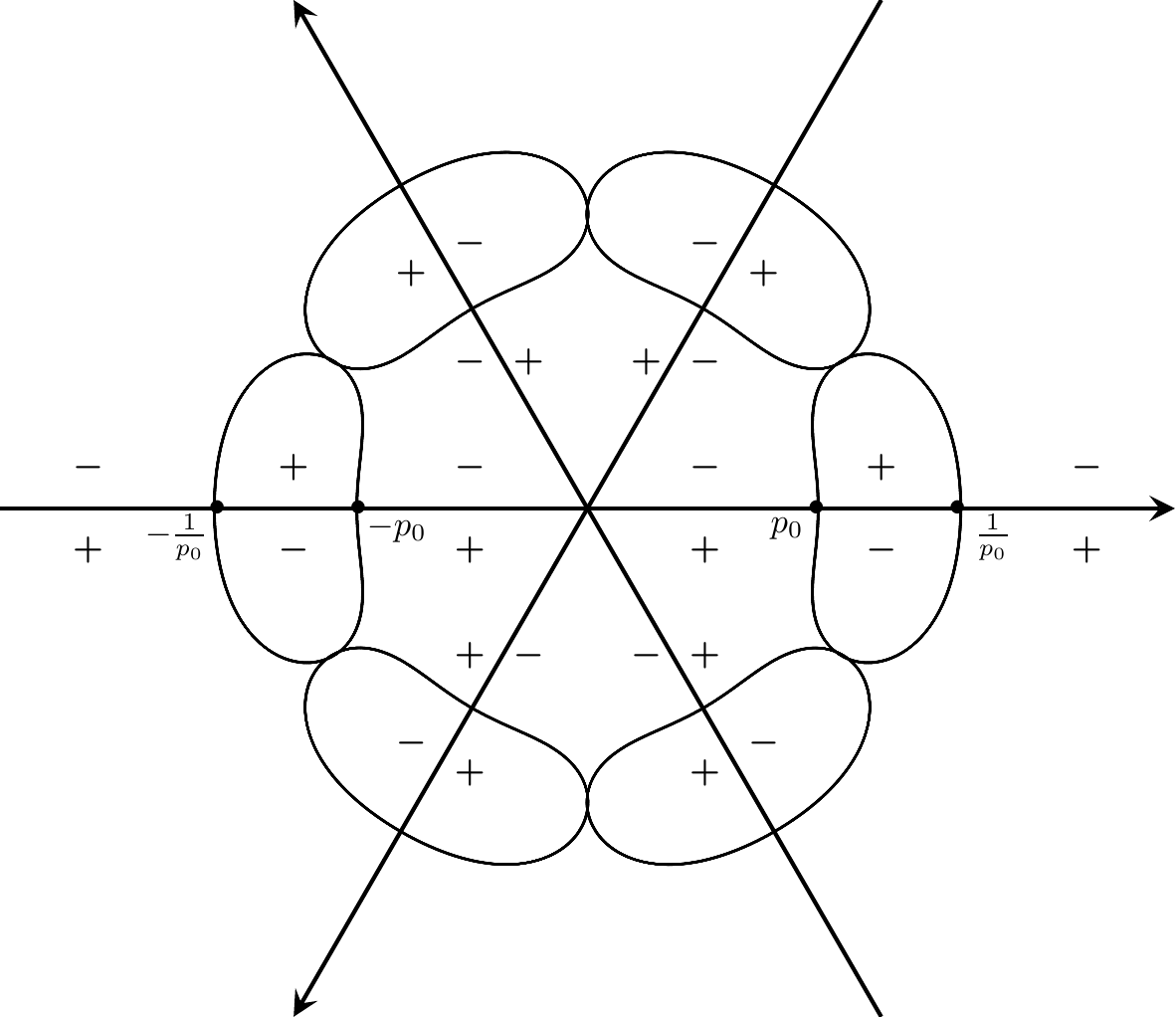}
\caption{Signs of $\Im\Theta(\xi,\tilde k(k))$ near each part (line) of $\Sigma$ for $0<\xi<2$} 
\label{fig:signs.parts.sigma}
\end{figure}

This signature table suggests introducing the diagonal factor $\tilde\delta(k)\equiv\tilde\delta(\xi,k)$:
\begin{align}   \label{tilde.delta}
\tilde\delta(k)&=
\begin{pmatrix}
\delta(k)\delta^{-1}(\omega^2k)&0&0\\
0&\delta^{-1}(k)\delta(\omega k)&0\\
0&0&\delta(\omega^2k)\delta^{-1}(\omega k)
\end{pmatrix}
\shortintertext{where}   \label{delta}
\delta(k)\equiv\delta(\xi,k)&=\exp\left\lbrace\frac{1}{2\ii\pi}\left(\int_{-\frac{1}{p_0(\xi)}}^{-p_0(\xi)}+\int_{p_0(\xi)}^{\frac{1}{p_0(\xi)}}\right)\frac{\log(1-\abs{r(s)}^2)^{-1}}{s-k}\,\dd s\right\rbrace,
\end{align}
in order to deform the RH problem to such a form that the associated jump matrix approaches, as $t\to\infty$, the identity matrix, uniformly outside small vicinities of the $12$ points: $-\frac{1}{p_0}$, $-p_0$, $p_0$, $\frac{1}{p_0}$ and their images after rotation by $\omega$ and by $\omega^2$. Indeed, introducing $\mu^{(1)}:=\mu\tilde\delta^{-1}(k)$, the jump conditions for $\mu^{(1)}$ are
\[
\mu_+^{(1)}=\mu_-^{(1)}S_1(\xi,t,k)
\]
with $S_1$ possessing appropriate triangular factorizations with non-diagonal terms decaying to $0$ upon deforming the contour from $\Sigma$ to a contour close to $\Sigma$ with self-intersection points at $-\frac{1}{p_0}$, $-p_0$, $p_0$, $\frac{1}{p_0}$, etc. For instance, for $k\in\D{R}$ (before the deformation) one has
\begin{subequations} \label{eq:jump}
\begin{align}
S_1(\xi,t,k)&=
\begin{pmatrix}
1&\frac{\bar r(k)}{1-\abs{r(k)}^2}\,\frac{\delta_-^2(k)}{\delta(\omega^2k)\delta(\omega k)}\,\ee^{-2\ii\tilde t\Theta(\xi,\tilde k(k))}&0\\
0&1&0\\
0&0&1
\end{pmatrix}\!\!\!
\begin{pmatrix}
1&0&0\\
\frac{r(k)}{1-\abs{r(k)}^2}\,\frac{\delta(\omega^2k)\delta(\omega k)}{\delta_+^2(k)}\,\ee^{2\ii\tilde t\Theta(\xi,\tilde k(k))}&1&0\\
0&0&1
\end{pmatrix}\notag\\
&\qquad\text{for }k\in(-\tfrac{1}{p_0},-p_0)\cup(p_0,\tfrac{1}{p_0})
\shortintertext{and}
S_1(\xi,t,k)&=
\begin{pmatrix}
1&0&0\\
-r(k)\frac{\delta(\omega^2k)\delta(\omega k)}{\delta^2(k)}\,\ee^{2\ii\tilde t\Theta(\xi,\tilde k(k))}&1&0\\
0&0&1
\end{pmatrix}\!\!\!
\begin{pmatrix}
1&-\bar r(k)\frac{\delta^2(k)}{\delta(\omega^2k)\delta(\omega k)}\,\ee^{-2\ii\tilde t\Theta(\xi,\tilde k(k))}&0\\
0&1&0\\
0&0&1
\end{pmatrix}\notag\\
&\qquad\text{for }k\in(-\infty,-\tfrac{1}{p_0})\cup(\tfrac{1}{p_0},\infty).
\end{align}
\end{subequations}
Similarly for $k\in\omega\,\D{R}$ and $k\in\omega^2\,\D{R}$.
\renewcommand{\qed}{}
\end{proof}

This has two important consequences:
\begin{enumerate}[(i)]
\item
rough estimate (see \cite{BS-D, BKST}): $\mu^{(1)}(\xi,t,k)=(1\ \ 1\ \ 1)+\ord(\frac{1}{\sqrt t})$ as $t\to\infty$;
\item
when analyzing the $\ord(\frac{1}{\sqrt t})$ term in more details, the main contribution comes from the sum of (separate) contributions of each small cross (after the contour deformation) centered at $-\frac{1}{p_0}$, $-p_0$, $p_0$, $\frac{1}{p_0}$ and their images after rotation by $\omega$ and by $\omega^2$ (cf.~\cite{DZ}).
\end{enumerate}
Now let us analyze the consequences of properties (i) and (ii).

\begin{prop}[symmetry] \label{p42}
We have $\abs{r(-k)}=\abs{r(k)}$, at least for $k\in\bigl(\frac{\sqrt{5}-1}{2},\frac{2}{\sqrt{5}-1}\bigr)$.
\end{prop}

\begin{proof}
Looking at \eqref{mu-0} we notice that $\mu_1(\,\cdot\,,\,\cdot\,,\ee^{\frac{\ii\pi}{6}})\mu_3(\,\cdot\,,\,\cdot\,,\ee^{\frac{\ii\pi}{6}})=\mu_2^2(\,\cdot\,,\,\cdot\,,\ee^{\frac{\ii\pi}{6}})$. On the other hand, property (i) above implies that $\mu_j(\,\cdot\,,t,\ee^{\frac{\ii\pi}{6}})=\tilde\delta_{jj}(\xi,\ee^{\frac{\ii\pi}{6}})\bigl(1+\ord(\frac{1}{\sqrt t})\bigr)$. It follows that $\tilde\delta_{11}(\xi,\ee^{\frac{\ii\pi}{6}})\tilde\delta_{33}(\xi,\ee^{\frac{\ii\pi}{6}})\tilde\delta_{22}^{-2}(\xi,\ee^{\frac{\ii\pi}{6}})=1$, hence by \eqref{tilde.delta}
\[
\delta^3(\xi,\ee^{\frac{\ii\pi}{6}})\delta^{-3}(\xi,\omega\ee^{\frac{\ii\pi}{6}})=1
\]
for all $\xi$ and thus is independent of $\xi$ (in the considered range of $\xi$).

On the other hand, from the definition \eqref{delta} of $\delta$ we have
\begin{equation}  \label{delta-relation}
\delta(\xi,\ee^{\frac{\ii\pi}{6}})\delta^{-1}(\xi,\ee^{\frac{5\ii\pi}{6}})=\exp\left\lbrace\frac{\ii\sqrt{3}}{2\pi}\left(\int_{-\frac{1}{p_0(\xi)}}^{-p_0(\xi)}+\int_{p_0(\xi)}^{\frac{1}{p_0(\xi)}}\right)\frac{\phi(s)}{s^2-\ii s-1}\,\dd s\right\rbrace,
\end{equation}
where $\phi(s):=\log(1-\abs{r(s)}^2)$. Noticing that $r(k)=\bar r(1/k)$ for $k\in\D{R}$ (this follows from property (S4) in Proposition \ref{p2}), then $\phi(s)=\phi(1/s)$ and the r.h.s.\ of \eqref{delta-relation} takes the form
\[
G(p_0(\xi)):=\exp\left\lbrace\int_{p_0(\xi)}^{\frac{1}{p_0(\xi)}}\frac{g(s)}{s^2-\ii s-1}\,\dd s\right\rbrace,
\]
where $g(s):=\frac{\ii\sqrt{3}}{2\pi}(\phi(s)-\phi(-s))$. Now, since the l.h.s.\ of \eqref{delta-relation} is independent of $\xi$, we have
\[
0=\frac{\dd}{\dd p_0}G(p_0)=-G(p_0)g(p_0)\frac{2\ii p_0}{(p_0^2-1)^2+p_0^2}\,.
\]
Therefore, $g(p_0)=0$ for all $p_0$ corresponding to $\xi\in(0,2)$, i.e., for $p_0\in(\frac{\sqrt{5}-1}{2},1)$, and thus, by symmetry, also for $p_0\in\bigl(\frac{\sqrt{5}-1}{2},\frac{2}{\sqrt{5}-1}\bigr)$, hence $\abs{r(s)}=\abs{r(-s)}$ for all $s\in\bigl(\frac{\sqrt{5}-1}{2},\frac{2}{\sqrt{5}-1}\bigr)$.
\end{proof}

Property (ii) implies that
\begin{equation}   \label{eq:mu-as}
\mu(\ee^{\frac{\ii\pi}{6}})=\begin{pmatrix}1&1&1\end{pmatrix}\biggl(\sum_{j=1}^{12}M^{(j)}(\ee^{\frac{\ii\pi}{6}})-11\cdot I\biggr)\tilde\delta(\ee^{\frac{\ii\pi}{6}})\bigl(I+\osmall(1)\bigr)\quad\text{as }t\to\infty,
\end{equation}
where $M^{(j)}(k)$ are the solutions of the $3\times 3$ matrix RH problems on the crosses centered at $-\frac{1}{p_0}$, $-p_0$, etc. 

Particularly, for $M^{(1)}(k)$ associated with the cross centered at $-\frac{1}{p_0}$, the factors in the jump matrix \eqref{eq:jump} can be approximated, as $t\to\infty$, as follows (cf.~\cite{BS-D,BKST}):
\begin{equation}   \label{approx.factors}
\delta^{\pm 2}(k)\,\delta^{\mp 1}(\omega^2k)\,\delta^{\mp 1}(\omega k)\,\ee^{\mp 2\ii\tilde t\Theta(\xi,\tilde k(k))}\approx\delta_{\ast}^{\pm 2}\cdot(-\hat k)^{\pm 2\ii h_0}\ee^{\mp\ii\hat k^2/2},
\end{equation}
where 
\[
h_0=-\frac{1}{2\pi}\log(1-\abs{r(p_0)}^2),
\]
and where the scaled spectral parameter $\hat k$ is defined by
\begin{align*}
\hat k&=-\bigl(\tilde k(k)-\kappa_0\bigr)\sqrt{c\tilde t}
\shortintertext{with}
c=c(\xi)&=\frac{32\kappa_0(\xi)(3-4\kappa_0^2(\xi))}{(1+4\kappa_0^2(\xi))^3}\,.
\end{align*}
Moreover, the constant (w.r.t.~$\hat k$) factor $\delta_{\ast}$ is as follows:
\begin{subequations}
\begin{align}
\delta_{\ast}&=(c\tilde t)^{-\frac{\ii h_0}{2}}\left(\frac{p_0}{1-p_0^2}\right)^{\ii h_0}\ee^{\ii\,\frac{16\kappa_0^3}{(1+4\kappa_0^2)^2}\,\tilde t}\,\ee^{\ii\chi_0},
\shortintertext{where}
\chi_0&=-\frac{1}{2\pi}\left(\int_{-\frac{1}{p_0}}^{-p_0}+\int_{p_0}^{\frac{1}{p_0}}\right)\log\left(s+\frac{1}{p_0}\right)\dd\log(1-\abs{r(s)}^2)\notag\\
&\quad-\frac{1}{2\pi}\left(\int_{-\frac{1}{p_0}}^{-p_0}+\int_{p_0}^{\frac{1}{p_0}}\right)\log(1-\abs{r(s)}^2)\,\frac{2s-\frac{1}{p_0}}{s^2-\frac{s}{p_0}+\frac{1}{p_0^2}}\,\dd s.
\end{align}
\end{subequations}
For $k$ near $-p_0$, the symmetry $\delta(k)=\overline{\delta(1/\bar k)}$ and \eqref{approx.factors} yield
\begin{equation}   \label{eq:approx.factors}
\delta^{\pm 2}(k)\,\delta^{\mp 1}(\omega^2k)\,\delta^{\mp 1}(\omega k)\,\ee^{\mp 2\ii\tilde t\Theta(\xi,\tilde k(k))}\approx\bar\delta_{\ast}^{\pm 2}\cdot\hat k^{\mp2\ii h_0}\ee^{\pm\ii\hat k^2/2},
\end{equation}
where now 
\[
\hat k=-\bigl(\tilde k(k)+\kappa_0\bigr)\sqrt{c\tilde t}.
\]
In turn, the symmetry $\delta(k)=\overline{\delta(-\bar k)}$ (which follows from Proposition \ref{p42}) and the property $\tilde k(-1/k) = \tilde k(k)$ imply that \eqref{approx.factors} holds also for the cross near $p_0$ while \eqref{eq:approx.factors} holds for the cross near $1/p_0$.

Conjugating out the constant factors in \eqref{approx.factors} and \eqref{eq:approx.factors}, the resulting problems on the crosses (in the $\hat k$ plane) become RH problems whose solutions are given in terms of parabolic cylinder functions \cite{DZ,BKST,BS-D}. Particularly, $M^{(1)}(k)\approx\Delta_{\one}\hat M^{(1)}\Delta_{\one}^{-2}$ with $\Delta_{\one}=\diag\accol{\delta_{\ast},\delta_{\ast}^{-1},1}$, where the large-$\hat k$ behavior of $\hat M^{(1)}(\hat k)$ is given by 
\begin{align*}
\hat M^{(1)}(\hat k)&=I+\frac{\hat M_1}{k}+\ord(\hat k^{-2}),
\shortintertext{where}
\hat M_1&=
\begin{pmatrix}
0&\ii\bar\beta&0\\
-\ii\beta&0&0\\
0&0&0\\
\end{pmatrix}
\shortintertext{with}
\beta&=\frac{r(-1/p_0)\Gamma(-\ii h_0)h_0}{\sqrt{2\pi}\,\ee^{\ii\pi/4}\ee^{-\pi h_0/2}}\,.
\end{align*}
Here $\Gamma$ is the Euler Gamma function.

Recalling the relationship $\hat k=-\bigl(\tilde k(k)-\kappa_0\bigr)\sqrt{c\tilde t}$, the evaluation of $M^{(1)}(\ee^{\ii\pi/6})$ as $t\to\infty$ reduces to the following:
\begin{subequations}  \label{eq:m-onetwo}
\begin{align} \label{M1}
M^{(1)}(\ee^{\ii\pi/6})&\simeq\Delta_{\one}\biggl(I-\frac{\hat M_1}{\sqrt{c\tilde t}\,\bigl(\tilde k(\ee^{\ii\pi/6})-\kappa_0\bigr)}\biggr)\Delta_{\one}^{-1}\notag\\
&=I+\frac{\check M_1}{\sqrt{c\tilde t}\,(\frac{\ii}{2}+\kappa_0)}\,,
\shortintertext{where}  \label{M1check}
\check M_1&=\Delta_{\one}\hat M_1\Delta_{\one}^{-1}
=\begin{pmatrix}
0&\ii\bar\beta\delta_{\ast}^2&0\\
-\ii\beta\delta_{\ast}^{-2}&0&0\\
0&0&0
\end{pmatrix}.
\shortintertext{Similarly,}\label{M2}
M^{(2)}(\ee^{\ii\pi/6})&\simeq\Delta_{\one}\left(I-\frac{-\overline{\hat M_1}}{\sqrt{c\tilde t}(\tilde k(\ee^{\ii\pi/6})+\kappa_0)}\right)\Delta_{\one}^{-1}\notag\\
&=I+\frac{\overline{\check M_1}}{\sqrt{c\tilde t}(-\frac{\ii}{2}+\kappa_0)}\,.
\end{align}
\end{subequations}
For $M^{(3)}$ and $M^{(4)}$, which correspond to the crosses at $p_0$ and $1/p_0$,  respectively, the scaled spectral parameters are $\hat k =-(\tilde k(k)-\kappa_0)\sqrt{c\tilde t}$ and $\hat k =-(\tilde k(k)+\kappa_0)\sqrt{c\tilde t}$, respectively. Hence, for $M^{(3)}(\ee^{\ii\pi/6})$ and $M^{(4)}(\ee^{\ii\pi/6})$ one has expressions similar to \eqref{M1} and \eqref{M2}, respectively, with $\beta$ replaced in \eqref{M1check} by
\[
\beta^*=\frac{r(p_0)\Gamma(-\ii h_0)h_0}{\sqrt{2\pi}\,\ee^{\ii\pi/4}\ee^{-\pi h_0/2}}\,.
\]

For the crosses centered along $\omega\,\D{R}$ at $-\omega/p_0$, $-\omega p_0$, $\omega p_0$ and $\omega/p_0$ (denote the solutions of the corresponding RH problems by $M^{(j)}$, $j=5,\dots,8$), one applies the symmetries of Proposition \ref{p2}, which gives the following. For $M^{(5)}(k)$ one has $M^{(5)}(k)\approx\Delta_{\two}\hat M^{(5)}(\hat k)\Delta_{\two}^{-2}$ with $\Delta_{\two}=\diag\accol{\delta_{\ast}^{-1},1,\delta_{\ast}}$, where now
\[
\hat k=-\bigl(\tilde k(\omega^2k)-\kappa_0\bigr)\sqrt{c\tilde t}.
\]
The large-$\hat k$ behavior of $\hat M^{(5)}(\hat k)$ is given by 
\begin{align*}
&\hat M^{(5)}(\hat k)=I+\frac{\hat M_5}{\hat k}+\ord(\hat k^{-2})
\shortintertext{with}
&\hat M_5=\begin{pmatrix}
0&0&-\ii\beta\\
0&0&0\\
\ii\bar\beta&0&0\\
\end{pmatrix}.
\end{align*}
Therefore,
\begin{subequations}  \label{eq:m-fivesix}
\begin{align}  \label{M5}
M^{(5)}(\ee^{\ii\pi/6})&\simeq\Delta_{\two}\left(I-\frac{\hat M_5}{\sqrt{c\tilde t}(\tilde k(\omega^2\ee^{\ii\pi/6})-\kappa_0)}\right)\Delta_{\two}^{-1}\notag\\
&\quad=I+\frac{\check M_5}{\sqrt{c\tilde t}(-\frac{\ii}{2}+\kappa_0)}\,,
\shortintertext{where}  \label{M5check}
\check M_5&=\Delta_{\two}\hat M_5\Delta_{\two}^{-1}
=\begin{pmatrix}
0&0&-\ii\beta\delta_{\ast}^{-2}\\
0&0&0\\
\ii\bar\beta\delta_{\ast}^2&0&0
\end{pmatrix}
\shortintertext{Similarly, for $k$ near $-\omega p_0$, one introduces $\hat k=-(\tilde k(\omega^2k)+\kappa_0)\sqrt{c\tilde t}$, which leads to} \label{M6}
M^{(6)}(\ee^{\ii\pi/6})&\simeq I+\frac{\overline{\check M_5}}{\sqrt{c\tilde t}(\frac{\ii}{2}+\kappa_0)}\,.
\end{align}
\end{subequations}
For $M^{(7)}$ and $M^{(8)}$, the scaled spectral parameters are $\hat k =-(\tilde k(\omega^2 k)-\kappa_0)\sqrt{c\tilde t}$ and $\hat k =-(\tilde k(\omega^2 k)+\kappa_0)\sqrt{c\tilde t}$, respectively, and the remark above concerning $M^{(3)}(\ee^{\ii\pi/6})$ and $M^{(4)}(\ee^{\ii\pi/6})$ is valid: $\beta$ in \eqref{M5check} is to be replaced by $\beta^*$.

For the crosses centered along $\omega^2\,\mathbb R$ at $-\omega^2/p_0$, $-\omega^2p_0$, $\omega^2p_0$, and $\omega^2/p_0$ (denote the solutions of the corresponding RH problems by $M^{(j)}$, $j=9,\dots,12$), introducing $\Delta_{\three}=\diag\accol{1,\delta_{\ast},\delta_{\ast}^{-1}}$, $\hat k = -(\tilde k(\omega k)-\kappa_0)\sqrt{c\tilde t}$ for $M^{(9)}$ and $M^{(11)}$, and $\hat k = -(\tilde k(\omega k)+\kappa_0)\sqrt{c\tilde t}$ for $M^{(10)}$ and $M^{(12)}$, one has the following:
\begin{subequations}  \label{eq:m-nineten}
\begin{align}
M^{(9)}(\ee^{\ii\pi/6})&\simeq\Delta_{\three}\left(I-\frac{\hat M_9}{\sqrt{c\tilde t}(\tilde k(\omega\ee^{\ii\pi/6})-\kappa_0)}\right)\Delta_{\three}^{-1}\notag\\
&=I+\frac{\check M_9}{\sqrt{c\tilde t}(\frac{\ii}{2}+\kappa_0)}\,,
\shortintertext{where}
\check M_9&=\Delta_{\three}\hat M_9\Delta_{\three}^{-1}
=\begin{pmatrix}
0&0&0\\
0&0&\ii\bar\beta\delta_{\ast}^2\\
0&-\ii\beta\delta_{\ast}^{-2}&0
\end{pmatrix}
\shortintertext{and}
M^{(10)}(\ee^{\ii\pi/6})&\simeq I+\frac{\overline{\check M_9}}{\sqrt{c\tilde t}(-\frac{\ii}{2}+\kappa_0)}\,,
\end{align}
\end{subequations}
whereas for $M^{(11)}$ and $M^{(12)}$, the remark above concerning $M^{(3)}$ and $M^{(4)}$ applies. One has the same expressions \eqref{eq:m-nineten} with $\beta$ replaced by $\beta^*$.

Collecting \eqref{eq:m-onetwo}-\eqref{eq:m-nineten} and substituting into \eqref{eq:mu-as} gives
\begin{align}   \label{mu123}
\mu_1(\ee^{\ii\pi/6})\tilde\delta_{11}^{-1}(\ee^{\ii\pi/6})&\simeq 1+\frac{4}{\sqrt{c\tilde t}}\,\Re\frac{-\ii(\beta+\beta^*)\delta_{\ast}^{-2}}{\frac{\ii}{2}+\kappa_0}\,,\notag\\
\mu_2(\ee^{\ii\pi/6})\tilde\delta_{22}^{-1}(\ee^{\ii\pi/6})&\simeq 1+\frac{2}{\sqrt{c\tilde t}}\left(\Re\frac{-\ii(\beta+\beta^*)\delta_{\ast}^{-2}}{-\frac{\ii}{2}+\kappa_0}+\Re\frac{-\ii(\beta+\beta^*)\delta_{\ast}^{-2}}{-\frac{\ii}{2}+\kappa_0}\right)\,,\notag\\
\mu_3(\ee^{\ii\pi/6})\tilde\delta_{33}^{-1}(\ee^{\ii\pi/6})&\simeq 1+\frac{4}{\sqrt{c\tilde t}}\,\Re\frac{-\ii(\beta+\beta^*)\delta_{\ast}^{-2}}{-\frac{\ii}{2}+\kappa_0}\,.
\end{align}
Taking into account the symmetry $\delta(k)=\bar\delta(1/\bar k)$, from \eqref{tilde.delta} and \eqref{delta} we can calculate
\begin{align}
\tilde\delta_{11}(\ee^{\ii\pi/6})&=\exp\left\lbrace\frac{1}{4\pi}\left(\int_{-\frac{1}{p_0}}^{-p_0}+\int_{p_0}^{\frac{1}{p_0}}\right)\log(1-\abs{r(s)}^2)\left(\frac{1}{s^2+s\sqrt{3}+1}+\frac{1}{s^2+1}\right)\dd s\right\rbrace,\notag\\
\tilde\delta_{22}(\ee^{\ii\pi/6})&=1,\notag\\
\tilde\delta_{33}(\ee^{\ii\pi/6})&=\tilde\delta_{11}^{-1}(\ee^{\ii\pi/6}).
\end{align}
Thus, from \eqref{mu123} we have
\begin{equation}   \label{mu21}
\log\frac{\mu_2}{\mu_1}(\ee^{\ii\pi/6})\simeq-\log\tilde\delta_{11}(\ee^{\ii\pi/6})+\frac{8}{\sqrt{c\tilde t}(1+4\kappa_0^2)}\,\Re\accol{(\beta+\beta^*)\delta_{\ast}^{-2}}.
\end{equation}
Recalling the definitions of $c$, $\beta$ and $\delta_{\ast}$, and taking into account that $\abs{\beta}=\abs{\beta^*}=\sqrt{h_0}$, we get
\begin{equation}   \label{rel}
\frac{8}{\sqrt{c\tilde t}(1+4\kappa_0^2)}\,\Re\accol{(\beta+\beta^*)\delta_{\ast}^{-2}}=\frac{\tilde c_1}{\sqrt{t}}\,\cos(c_2t+c_3\log t+\tilde c_4),
\end{equation}
where
\begin{subequations}   \label{eq:tc14}
\begin{align}
\tilde c_1&=\left(\frac{16 h_0(1+4\kappa_0^2)}{3\kappa_0(3-4\kappa_0^2)}\right)^{\frac{1}{2}}\cos\bigl(\arg r(-p_0^{-1})-\arg r(p_0)\bigr)\\
c_2&=\frac{48\kappa_0^3}{(1+4\kappa_0^2)^2}\,,\qquad c_3=-h_0,\\
\tilde c_4&=\frac{\pi}{4}+\arg\Gamma(\ii\,h_0)-h_0\log\frac{48\kappa_0(3-4\kappa_0^2)(1-p_0^2)^2}{(1+4\kappa_0)^3p_0^2}\notag\\
&\quad-\frac{\arg r(-1/p_0)+\arg r(p_0)}{2}-2\chi_0.
\end{align}
\end{subequations}
Now the asymptotics of $u(x,t)$ can be calculated by differentiating \eqref{mu21}, \eqref{rel} with respect to $t$ (keeping $y$ fixed), setting $\kappa_0=\kappa_0(\xi)$ \eqref{kappa-0} in terms of $\xi=\frac{2y}{3t}$, so that $\frac{2y}{3t}=\frac{2(1-4\kappa_0^2)}{(1+4\kappa_0^2)^2}$ by \eqref{eq:tildek.critical}, and taking into account the change of variables $y\mapsto x$. Due to
\[
y=x+\log\tilde\delta_{11}(\ee^{\ii\pi/6})+\ord(t^{-1/2}),
\]
the latter results in an additional phase shift
\begin{align}
\Delta&=-2\kappa_0\log\tilde\delta_{11}(\ee^{\ii\pi/6})\notag\\
&=-\frac{\kappa_0}{2\pi}\left(\int_{-\frac{1}{p_0}}^{-p_0}+\int_{p_0}^{\frac{1}{p_0}}\right)\log(1-\abs{r(s)}^2)\left(\frac{1}{s^2+s\sqrt{3}+1}+\frac{1}{s^2+1}\right)\dd s.
\end{align}
The resulting formula is
\begin{equation}
u(x,t)=\frac{c_1}{\sqrt{t}}\,\sin(c_2t+c_3\log t+c_4)(1+\osmall(1)),
\end{equation}
where
\begin{subequations}   \label{eq:c14}
\begin{align}
c_1&=-\tilde c_1\left(c_2+\frac{\partial c_2}{\partial t}\right)=-\tilde c_1\left(c_2+\frac{\dd c_2}{\dd\kappa_0}\frac{\partial\kappa_0}{\partial t}\right)=-\tilde c_1\frac{6\kappa_0}{1+4\kappa_0^2}\notag\\
&=-\left(\frac{192 h_0\kappa_0}{(3-4\kappa_0^2)(1+4\kappa_0^2)}\right)^{\frac{1}{2}}\cos\bigl(\arg r(-p_0^{-1})-\arg r(p_0)\bigr),\\
c_2&=\frac{48\kappa_0^3}{(1+4\kappa_0^2)^2}\,,\qquad c_3=-h_0,\\
c_4&=\tilde c_4+\Delta.
\end{align}
\end{subequations}

\subsection{Asymptotics in other sectors}   \label{asy.other.sectors}

The asymptotics in other sectors as well as in transition zones connecting the sectors is determined by the corresponding signature tables which dictate appropriate deformations and rescalings. 

\subsubsection{Sector \emph{(iii)}}   \label{asy.sector.three}

For $-\frac{3}{8}<\frac{x}{t}<0$, i.e., for $-\frac{1}{4}<\xi<0$ the critical point equation \eqref{eq:varpi} has two solutions $\geq 1$, which gives in the $\tilde k$-plane four real critical points $\pm\kappa_0$, $\pm\kappa_1$:
\begin{subequations}   \label{kappa01}
\begin{align}   \label{kappa0}
\kappa_0(\xi)&=\left(\frac{\sqrt{1+4\xi}-1-\xi}{4\xi}\right)^{\frac{1}{2}},\\
\label{kappa1}
\kappa_1(\xi)&=\left(-\frac{\sqrt{1+4\xi}+1+\xi}{4\xi}\right)^{\frac{1}{2}}.
\end{align}
\end{subequations}
Applying the inverse map $\tilde k\leadsto k$ we get in the $k$-plane eight real critical points $\pm p_0$, $\pm\frac{1}{p_0}$, $\pm p_1$, $\pm\frac{1}{p_1}$:
\begin{subequations}   \label{p01}
\begin{align}   \label{p0}
p_0(\xi)&=-\kappa_0+\sqrt{\kappa_0^2+1}\,,\\
\label{p1}
p_1(\xi)&=-\kappa_1+\sqrt{\kappa_1^2+1}\,.
\end{align}
\end{subequations}
Thus, the signs of $\Im\Theta(\xi,\tilde k(k))$ near the real axis are distributed as seen in Figure~\ref{fig:signs.two}. There are eight critical points on $\D{R}$, and similarly, on $\omega\D{R}$ and $\omega^2\D{R}$. Eventually (as for the Camassa--Holm equation), this leads to the sum of two oscillating modes in the asymptotics.
\begin{figure}[ht]
\centering\includegraphics[scale=.7]{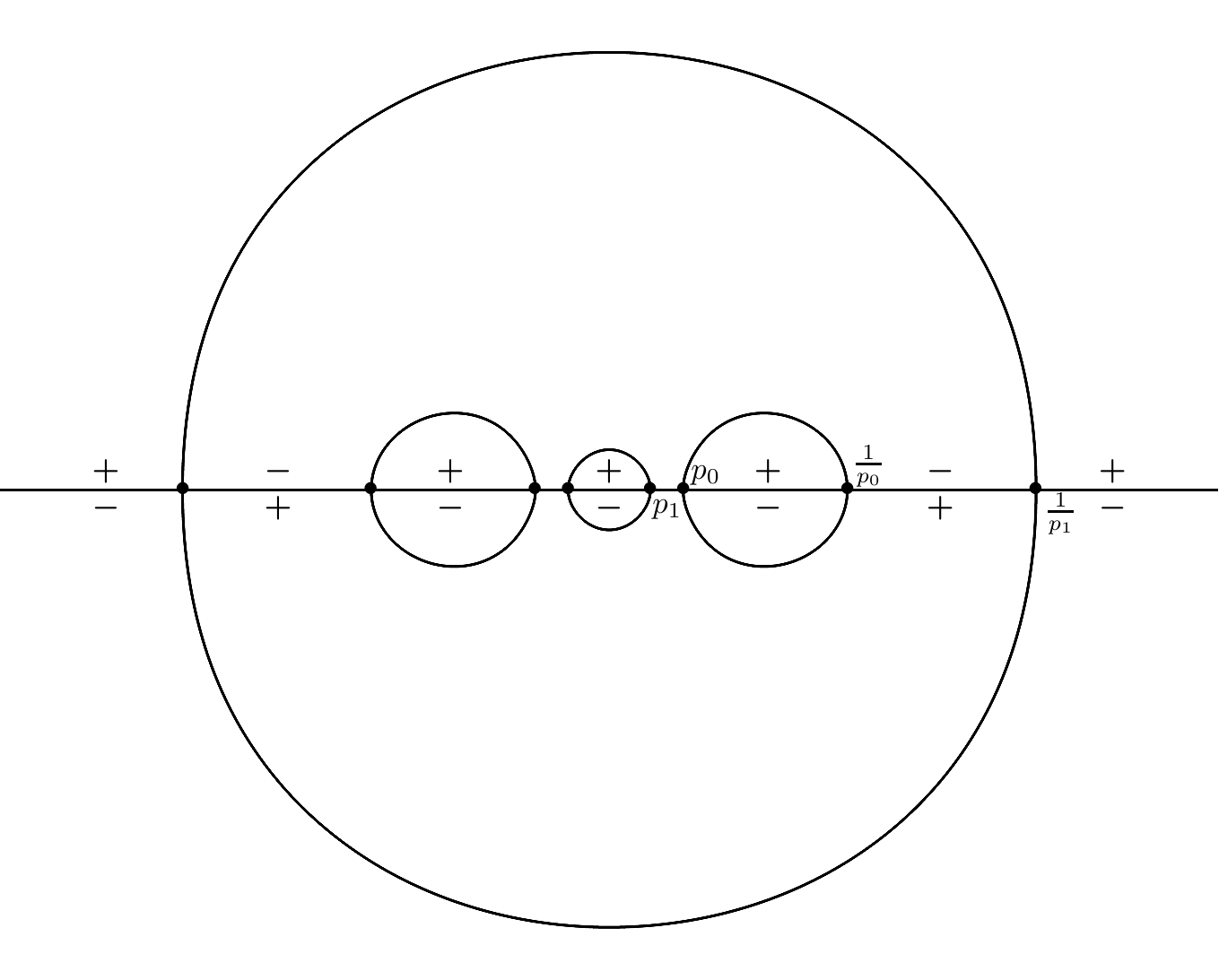}
\caption{Signs of $\Im\Theta(\xi,\tilde k(k))$ near $\D{R}$ for $-\frac{3}{8}<\frac{x}{t}<0$, i.e., $-\frac{1}{4}<\xi<0$.} 
\label{fig:signs.two}
\end{figure}

\subsubsection{Transition zone \emph{(iii)-(iv)}}\label{asy.transition.three.four}

At $\frac{x}{t}=-\frac{3}{8}$, i.e., at $\xi=-\frac{1}{4}$, the critical equation  \eqref{eq:varpi} has a solution of multiplicity $2$, which gives $\kappa_0=\kappa_1=\frac{\sqrt{7}-\sqrt{3}}{2}$. Thus, in the $k$-plane pairs of critical points collide: $p_0=p_1=\frac{\sqrt{7}-\sqrt{3}}{2}$, $\frac{1}{p_0}=\frac{1}{p_1}=\frac{\sqrt{7}+\sqrt{3}}{2}$, as seen in Figure~\ref{fig:transition.three.four}, 
\begin{figure}[ht]
\centering\includegraphics[scale=.7]{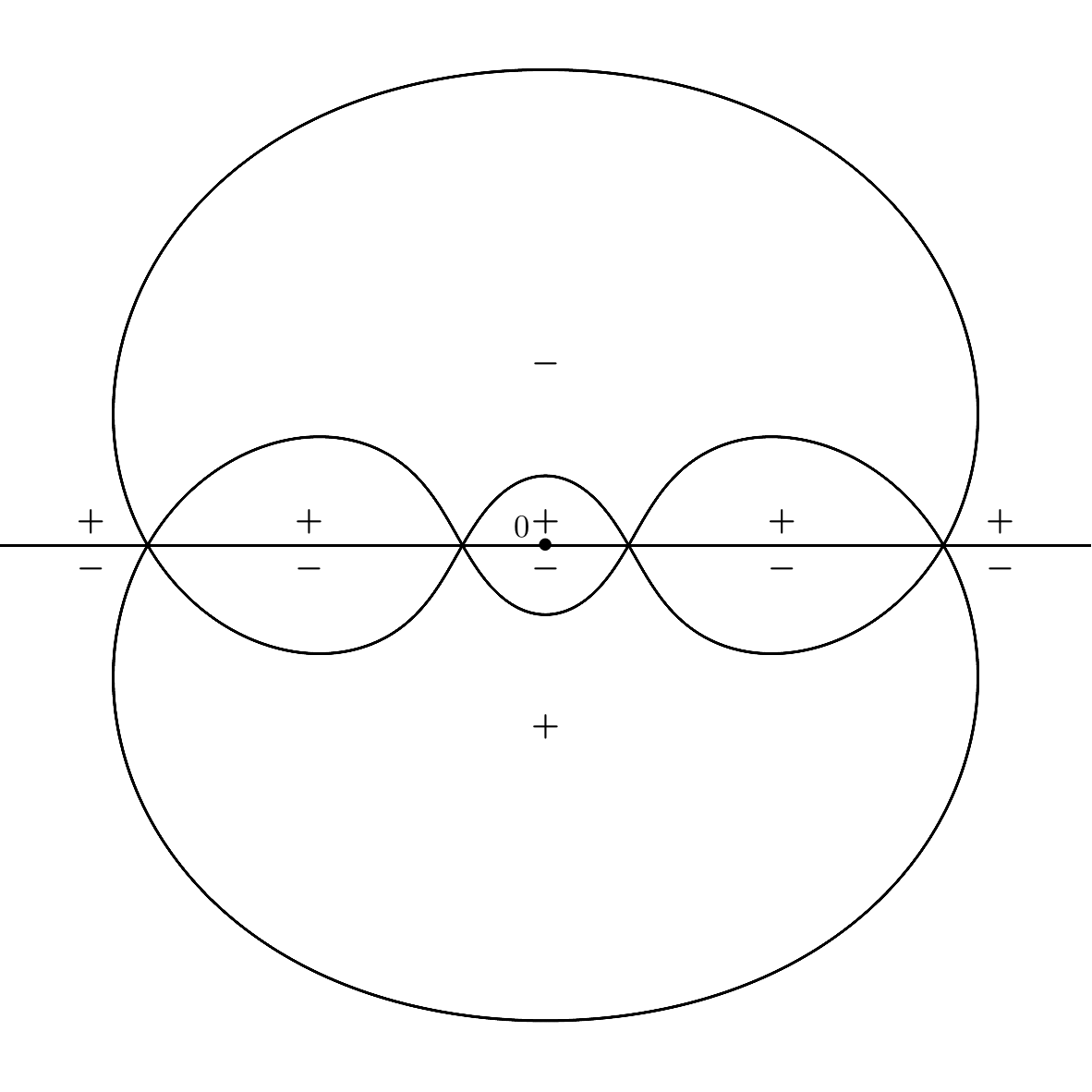}
\caption{Signs of $\Im\Theta(\xi,\tilde k(k))$ in the $k$-plane for $\frac{x}{t}=-\frac{3}{8}$, i.e., $\xi=-\frac{1}{4}$.} 
\label{fig:transition.three.four}
\end{figure}
which signifies the occurrence of a transition zone between the sectors (iii) and (iv), where the main asymptotic terms are expressed in terms of dedicated solutions of the Painlev\'e II equation (cf.~[3]). 

\subsubsection{Transition zone \emph{(i)-(ii)}}\label{asy.transition.one.two}

At $\frac{x}{t}=3$, i.e., at $\xi=2$, the only positive solution of \eqref{eq:varpi} is $\varpi=1$, hence $\kappa_0=0$ and pairs of critical points collide: $p_0=\frac{1}{p_0}=1$. There is a transition between sectors (i) and (ii) (see Figure~\ref{fig:transition.one.two}) which is similar to the previous one.
\begin{figure}[ht]
\centering\includegraphics[scale=1]{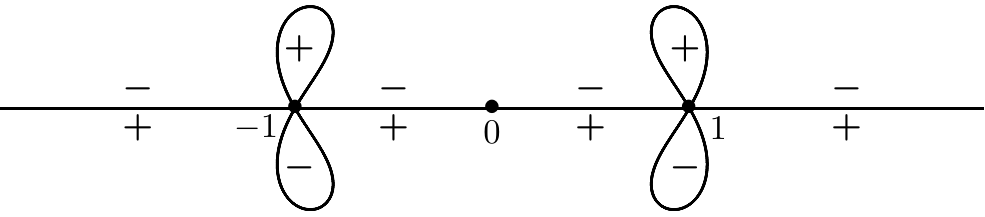}
\caption{Signs of $\Im\Theta(\xi,\tilde k(k))$ in the $k$-plane for $\frac{x}{t}=3$, i.e., $\xi=2$.} 
\label{fig:transition.one.two}
\end{figure}

\subsubsection{Sector \emph{(iv)}}   \label{asy.sector.four}

For $\frac{x}{t}<-\frac{3}{8}$, i.e., for $\xi<-\frac{1}{4}$, the critical equation \eqref{eq:varpi} has no real solution. Thus there are no real critical points as seen in Figures~\ref{fig:signs.fourth.sector.one}-\ref{fig:signs.fourth.sector.three} corresponding to sector (iv), which leads to a rapid decay of the solution in this sector.
\begin{figure}[ht]
\centering\includegraphics[scale=.6]{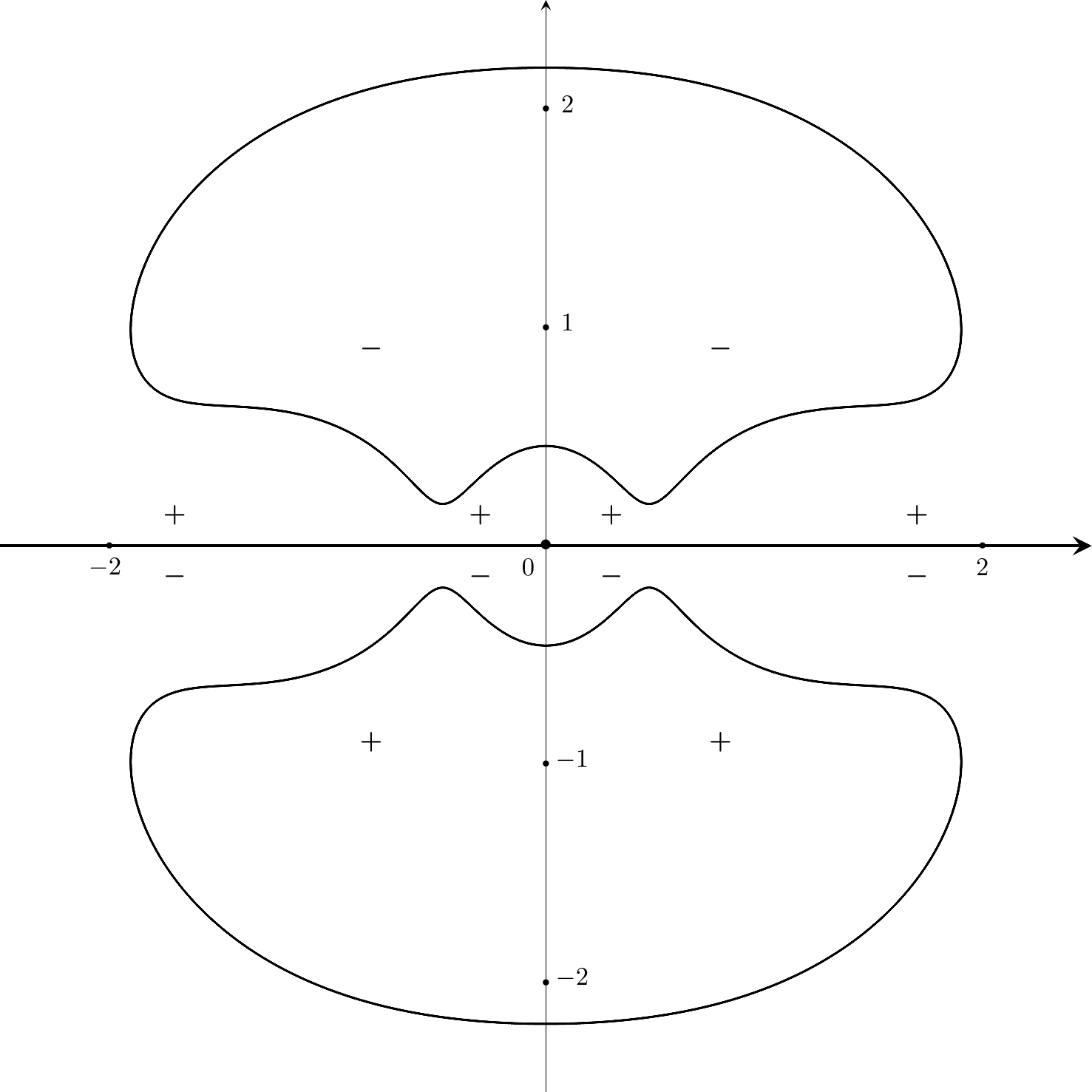}
\caption{Signs of $\Im\Theta(\xi,\tilde k(k))$ for $-1<\frac{x}{t}<-\frac{3}{8}$, i.e., $-\frac{2}{3}<\xi<-\frac{1}{4}$.} 
\label{fig:signs.fourth.sector.one}
\end{figure}

\begin{figure}[ht]
\centering\includegraphics[scale=.7]{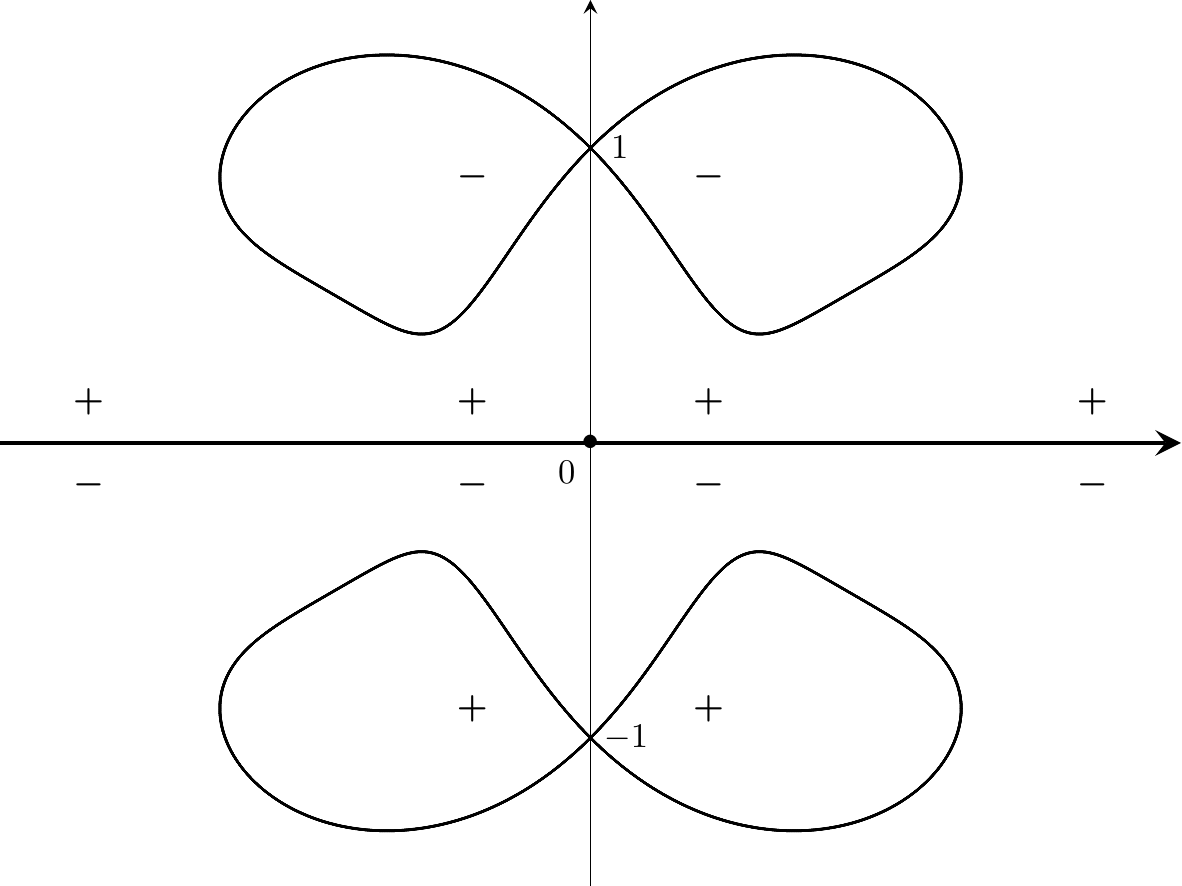}
\caption{Signs of $\Im\Theta(\xi,\tilde k(k))$ in the $k$-plane for $\frac{x}{t}=-1$, i.e., $\xi=-\frac{2}{3}$.} 
\label{fig:signs.fourth.sector.two}
\end{figure}

\begin{figure}[ht]
\centering\includegraphics[scale=.7]{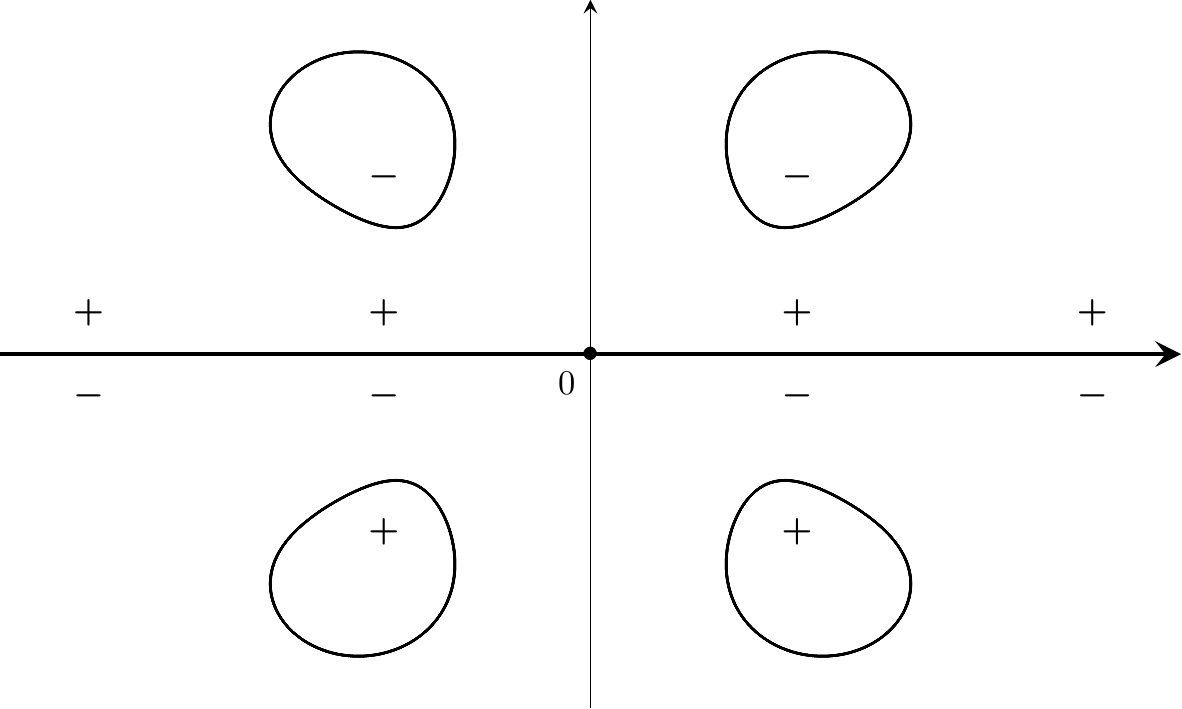}
\caption{Signs of $\Im\Theta(\xi,\tilde k(k))$ in the $k$-plane for $\frac{x}{t}<-1$, i.e., $\xi<-\frac{2}{3}$.} 
\label{fig:signs.fourth.sector.three}
\end{figure}

\subsubsection{Sector \emph{(i)}}   \label{asy.sector.one}

For $\frac{x}{t}>3$, i.e., for $\xi>2$, the critical equation \eqref{eq:varpi} has no solution $\geq 1$ and thus there are no real critical points, which leads to a rapid decay of the solution in this sector.


\end{document}